\documentclass{article}

\usepackage{graphicx} % Required for inserting images
\usepackage{amsmath}   % binom
\usepackage{amsfonts}  % mathbb
\usepackage{amssymb}   % rtimes
\usepackage{amsthm}    % proof environment
\usepackage{mathtools} % coloneqq
\usepackage{dsfont}    % identity operator
\usepackage{thm-restate}
\usepackage[dvipsnames]{xcolor}
\usepackage{hyperref}
\usepackage{caption}
\usepackage{subcaption}
\usepackage{authblk}
\usepackage[a4paper, margin=1in]{geometry}

\usepackage[backend=biber, style=numeric, sorting=none,url=false,isbn=false,doi = false,date = year, eprint=true,]{biblatex}
\newbibmacro{string+doi}[1]{
  \iffieldundef{doi}{#1}{\href{https://doi.org/\thefield{doi}}{#1}}}
\DeclareFieldFormat{title}{\usebibmacro{string+doi}{#1}}
\DeclareFieldFormat[article,thesis]{title}{\usebibmacro{string+doi}{\mkbibquote{#1}}}

\DeclareMathOperator{\tr}{tr}
\DeclareMathOperator{\re}{Re}
\DeclareMathOperator{\im}{Im}
\newcommand{\trb}[1]{\tr \big( #1 \big)}

\newcommand{\bxi}{\boldsymbol{\xi}}
\newcommand{\bsigma}{\boldsymbol{\sigma}}
\newcommand{\bmu}{\boldsymbol{\mu}}
\newcommand{\boeta}{\boldsymbol{\eta}}

\newcommand{\CL}{\mathcal{L}}
\newcommand{\swap}{\texttt{SWAP}}
\newcommand{\CLp}{\mathcal{L}^{\perp}}
\newcommand{\CD}{\mathcal{D}}
\newcommand{\CC}{\mathcal{C}}

\newcommand{\CE}{\mathcal{E}}
\newcommand{\CN}{\mathcal{N}}

\newcommand{\C}{\mathbb{C}}
\newcommand{\R}{\mathbb{R}}
\newcommand{\N}{\mathbb{N}}
\newcommand{\Z}{\mathbb{Z}}
\newcommand{\bA}{\mathbf{A}}
\newcommand{\bB}{\mathbf{B}}
\newcommand{\norm}[1]{\lVert {#1} \rVert}

\newcommand{\id}{\mathds{1}}
\newcommand{\ket}[1]{|{#1} \rangle}
\newcommand{\bra}[1]{\langle {#1} |}

\newcommand{\ketbra}[1]{| {#1} \rangle \! \langle {#1}|}

\newcommand{\err}{\texttt{error}}

\newcommand{\Bigmid}{\, \Big\vert \,}
\newcommand{\Biggmid}{\, \Bigg\vert \,}

\newcommand{\dEcrit}{3.4286}
\newcommand{\dLcrit}{4.9193}

\newcommand*\diff{\mathop{}\!\mathrm{d}}

\usepackage{thmtools,thm-restate}

\newtheorem{definition}{Definition}
\newtheorem{lemma}{Lemma}
\newtheorem{assumption}{Assumption}

\newtheorem{theorem}{Theorem}

\title{Continuous\kern0.05em-\kern-0.1em Variable\\ Quantum MacWilliams Identities}
\author[1]{Ansgar G. Burchards \thanks{ansgarburchards@gmail.com}}
\affil[1]{\small{Dahlem Center for Complex Quantum Systems, Physics Department, Freie
Universit{\"a}t Berlin, Arnimallee 14, 14195 Berlin, Germany}}
\date{\vspace{-5ex}}
\begin{document}
\maketitle

\begin{abstract}
We derive bounds on general quantum error correcting codes against the displacement noise channel. The bounds limit the distances attainable by codes and also apply in an approximate setting. Our main result is a quantum analogue of the classical Cohn-Elkies bound on sphere packing densities attainable in Euclidean space. We further derive a quantum version of Levenshtein's sphere packing bound and argue that Gottesman--Kitaev--Preskill (GKP) codes based on the $E_8$ and Leech lattices achieve optimal distances. The main technical tool is a continuous-variable version of the quantum MacWilliams identities, which we introduce. The identities relate a pair of weight distributions which can be obtained for any two trace-class operators. General properties of these weight distributions are discussed, along with several examples.
\end{abstract}

\tableofcontents

\section{Introduction}
\label{sec:introduction}
Weight distributions play an important role in the theory and analysis of error-correcting codes. In the classical linear setting, they convey strictly more information than a code’s size $k$
and distance $d$, instead fully characterizing the distribution of weights of undetectable errors. 
A central result concerning weight distributions are the MacWilliams identities, considered by Van Lint one of the most fundamental results in coding theory~\cite{vanlint_1982_Introduction}.
As originally derived, they relate the weight distribution of a classical linear error-correcting code to that of its dual code~\cite{macwilliams_1963_theorem}, with a later generalization applying also to nonlinear codes~\cite{macwilliams_1972_MacWilliams}.

In the theory of quantum coding, weight distributions for both stabilizer and non-stabilizer codes were introduced by Shor and Laflamme~\cite{shor_1997_Quantum}. They associate to each quantum code two weight distributions, which contain information about the susceptibility of the code to Pauli errors of any weight. As in the classical case, these quantum weight distributions are related by a linear transformation, which constitutes a quantum version of the MacWilliams identities.

Notably, constraints on weight distributions in both the classical and quantum case can be leveraged in order to derive bounds on code parameters. When combined with the MacWilliams identities, such constraints give rise to upper bounds on the sizes of codes. These bounds take the form of a linear program for every set, $n, k,$ and $d$, of code parameters, whose infeasibility implies the nonexistence of a code with the given parameters. Linear programming bounds are among the tightest known general bounds, both for finite parameters as well as asymptotically~\cite{shor_1997_Quantum, delsarte_1972_Bounds, mceliece_1977_New, ashikhmin_1999_Upper}.
 
Linear programming bounds also play an important role in the classical sphere packing problem -- determining the maximum fraction of $n$-dimensional space that can be covered by  non-overlapping unit balls. The sphere packing problem can be regarded as a continuous-variable analogue of the coding problem, and the corresponding linear programming bound was formulated by Cohn and Elkies~\cite{cohn_2003_New}. 
Given an appropriate auxiliary function, $f:\R^{n} \rightarrow \R$, their bound provides an upper limit on the density of sphere packings achievable in $n$-dimensional space. The bound not only reproduces the tightest known upper bound on sphere packing densities in high dimensions~\cite{cohn_2014_Sphere}, but, remarkably, is tight in dimensions $8$ and $24$. This was shown in dimension $8$ by Viazovska through the construction of an appropriate ‘magic’ auxiliary function, whose resulting density upper bound matches the packing density achieved by the $E_8$ lattice~\cite{viazovska_2017_sphere}. Subsequent work resulted in the construction of an analogous function, which establishes the Leech lattice as an optimal sphere packing in dimension $24$~\cite{cohn_2023_sphere}.

In this work, we introduce new weight distributions and the corresponding MacWilliams identities for the continuous-variable quantum setting. Specifically, for each pair of trace class operators, $\hat{O}_1$ and $\hat{O}_2$,  on the Hilbert space containing the quantum state of $N$ linearly constrained degrees of freedom (modes), we construct a pair of weight distributions $\bA, \bB: \R_{\geq 0} \rightarrow \C$. Both weight distributions are related by a certain linear integral transformation, a continuous-variable version of the quantum MacWilliams identities. After establishing some general properties of these distributions, we study in detail the weight distributions of quantum error correcting codes, that is, the case $\Pi = \hat{O}_1 = \hat{O}_2$, with $\Pi$ the projector onto a finite dimensional subspace of Hilbert space.
Quantifying the error correcting capabilities of a code by its distance---the length of phase-space displacements which are guaranteed to be detectable---we show that the distance of an error-correcting code is reflected in its weight distributions. More specifically, we show that $\bA(r) = \bB(r)$ for all $r$ exceeded by the distance. Note that our approach is not obtained from that of Shor-Laflamme in the limit of diverging local qudit dimension, where the severity of an error is measured by the size of its support instead of its length in phase-space.

In order to enhance the generality of our work, we also introduce a notion of approximate quantum error detecting code (QEDC) of quality $\epsilon$. This notion reduces to the definition of a quantum error correcting code (QECC) in the sense of Knill and Laflamme in the ideal case, $\epsilon=0$, and relaxes their error correction conditions whenever $\epsilon >0$. The parameter $\epsilon$ also has a clear operational interpretation: it is the failure rate of a code in a task where errors are to be detected on an entangled state, only part of which is supported within the code space.

We employ the continuous-variable quantum MacWilliams identities in order to derive a quantum version of the Cohn-Elkies sphere packing bound. This quantum version of the bound limits the size $K$ of a QEDC in terms of its distance $d$ and the number of modes $N$ it is supported on. In the special case that $\epsilon = 0$, it then constrains the sizes of QECCs against the displacement noise channel in the sense of Knill and Laflamme. The bound employs an auxiliary radial function, $f\colon\R^{2N} \rightarrow \R$, that is, a function whose value depends only on the norm of its argument.

\begin{restatable}[Quantum Cohn-Elkies bound]{theorem}{QuantumCohnElkies}
\label{thm:quantum_cohn_elkies}
    Let $\hat{f}\colon \R^{2N} \rightarrow \R$ be a bounded, nonzero, non-negative radial function whose bounded Fourier transform satisfies $f(x) \geq 0$ for $x < d$ and $f(x) \leq 0$ for $x \geq d$, then the parameters of any $[[N, K, d, \epsilon]]$-QEDC satisfy the inequality
    \begin{equation}
    K \leq \frac{1}{1-\epsilon}\sup\Bigg\{\frac{f(x)}{\hat{f}(x)} \Biggmid x \in [0, d] \Bigg\} \, .
\end{equation}
\end{restatable}

The bound provides an upper limit on the size of a continuous-variable code for every auxiliary function $f$ satisfying a set of linear constraints. Optimizing the bound over such auxiliary functions is, in general, intractable. To obtain a more concrete bound in terms of code parameters, we instead apply the theorem to an appropriate, non-optimal family of auxiliary functions originally obtained by Cohn and Elkies based on calculus of variations arguments~\cite{cohn_2003_New}. This yields the following quantum version of Levenshtein's sphere packing bound.

\begin{restatable}[Quantum Levenshtein bound]{theorem}{QuantumLevenshtein}
\label{thm:quantum_levenshtein}
    For $0 < d \leq d_{+}$ any $[[N, K, d, \epsilon]]$-QEDC must satisfy the inequality \begin{equation}
        K d^{2N}\leq \frac{1}{1-\epsilon} \frac{j_{N}^{2 N}}{N! 2^N}
    \end{equation}
    where $d_{+}$ is as in Lemma~\ref{lem:levenshtein_supremum}.
\end{restatable}
Here, $j_{N}$ denotes the first positive zero of the Bessel function of the first kind $J_{N}$. Note that this version of the Levenshtein bound only applies to distances below some threshold $d_{+}$, for which we provide a closed form expression. In particular, the bound implies that, for any fixed encoded logical dimension $K$, the distance can grow no faster than $\mathcal{O}(\sqrt{N})$.

Finally, we apply Theorem~\ref{thm:quantum_cohn_elkies} to the ‘magic’ auxiliary functions $f_8$ and $f_{24}$, which led to the resolution of the $8$- and $24$-dimensional sphere packing problems. Conditional on an additional assumption on the respective function, we conclude that the distances 
 achieved by ideal ($\epsilon = 0$) Gottesman-Kitaev-Preskill (GKP) codes based on the $E_8$ and Leech lattices cannot be exceeded, even by non-lattice constructions. The additional assumptions concern the maxima achieved by quotients $f_{8}(\sqrt{2} x)/\hat{f}_{8}(\alpha x)$ and $f_{24}(2 x)/\hat{f}_{24}(\alpha x)$ on the unit interval for certain values of $\alpha >0$. While we do not provide formal proofs of the assumptions, we verify their validity by numerically evaluating the quotients and providing plots over the unit interval.

\begin{restatable}[Optimality of $E_8$ and Leech GKP codes]{theorem}{QuantumEEightLeechOptimality}
Assumption~\ref{assumption:f8_sup} implies that for $d \leq d_{8}^{\textnormal{(max)}} \approx \dEcrit$ any $[[4, K , d, \epsilon]]$-QEDC must satisfy the inequality 
    \begin{equation}
    \label{eq:e8bound}
        Kd^{8} \leq  \frac{(4\pi)^4}{1-\epsilon}\, .
    \end{equation}
Assumption~\ref{assumption:f24_sup} implies that for $d \leq d_{24}^{\textnormal{(max)}}\approx \dLcrit$ any $[[12, K , d, \epsilon]]$-QEDC must satisfy the inequality 
    \begin{equation}
    \label{eq:e24bound}
        Kd^{24} \leq \frac{(8\pi)^{12}}{1-\epsilon}\, .
    \end{equation}
\end{restatable}
Both of these bounds extend over the full distance range achievable by error correcting codes. That is, both upper limits $d_{8}^{\textnormal{(max)}}$ and $d_{24}^{\textnormal{(max)}}$ exceed the critical distances where Eqs.~\eqref{eq:e8bound} and~\eqref{eq:e24bound}, respectively, achieve equality for the smallest nontrivial code size $K = 2$. The upper limits are hence sufficient for the bounds to constrain all error-correcting codes, as well as to provide lower bounds on the code quality $\epsilon$ of QEDCs for distances exceeding those achievable in the ideal case when $\epsilon =0$.

\section{Background}
\label{sec:background}
In this section we fix some notational conventions and review background information on the Fourier transformation as well as continuous-variable quantum systems. We refer the reader to Refs.~\cite{cohn_2016_Packinga, sloane_1998_sphere, conway_1999_Sphere} for additional background on the sphere packing problem and Refs.~\cite{serafini_2017_Quantuma,weedbrook_2012_Gaussian} for more information on continuous-variable quantum systems. Some introductions to continuous-variable quantum error correction can be found in Refs.~\cite{albert_2022_Bosonic, brady_2024_Advances, terhal_2020_Scalable} and a mathematically oriented discussion of lattice based codes has been given in Ref.~\cite{conrad_2022_GottesmanKitaevPreskill}.

\subsection{Fourier transformation}
The unitary Fourier transformation on functions $f:\R^{2N} \rightarrow \C$ is defined as
\begin{equation}
\label{eq:fourier_definition}
    \hat{f}(\bold{y}) = \frac{1}{(2 \pi)^{N}}\int \diff \bold{x}\, e^{-i\bold{y}^{T}\bold{x}} f(\bold{x}) \, .
\end{equation}
Throughout we will be primarily concerned with radial functions, that is those functions $f$ whose values only depend on the magnitude of their argument. We will denote vectors in bold face and occasionally denote their magnitude in normal font, i.e.\ $x = \norm{\bold{x}}$. Hence, a radial function is one that satisfies $f(\bold{x}) = f(x)$ for some $f: \R_{\geq 0} \rightarrow \C$, which we denote by the same symbol as the original function. Note that if $f$ is a radial function then so is $\hat{f}$. As a consequence one can explicitly write the Fourier transform~\eqref{eq:fourier_definition} as a transformation between the respective radial representations. Concretely,
\begin{equation}
\label{eq:radial_fourier_transform_background}
    \hat{f}(y) = \frac{1}{y^{N-1}}\int_{0}^{\infty} \diff r\, J_{N-1}(yr) r^{N}f(r) \, ,
\end{equation}
where $J_{N}$ denotes the Bessel function of first kind. For a proof see e.g.\ Ref.~\cite{andrews_1999_Special}.

\subsection{Continuous-variable quantum systems}
In classical physics the state of $N$ continuous degrees of freedom is given by specifying their respective positions and momenta, a set of parameters conveniently described by a vector $\bold{x} \in \R^{2N}$. In the theory of continuous-variable quantum systems, both position and momentum of each degree of freedom must be upgraded to position operators $\hat{q}_i$ and momentum operators $\hat{p}_i$, which act on a separable Hilbert space. As in the classical case we define the ‘qqpp’-ordered vector of operators
\begin{equation}
    \hat{\bold{x}} = (\hat{q}_1, \dots, \hat{q}_N, \hat{p}_1, \dots , \hat{p}_N)^{T} \, .
\end{equation}
The elementary fact that position and momentum are conjugate variables, both of which can not be known arbitrarily precisely at a given time, is reflected in the nontrivial commutation relations between the respective operators.
In particular, position and momentum satisfy the well-known canonical commutation relations
\begin{equation}
\label{eq:canonical_commutation_relations}
    [\hat{\bold{x}}_i , \hat{\bold{x}}_j] = i \Omega_{ij}
\end{equation}
where 
\begin{equation}
    \Omega = \begin{pmatrix}
        0 & \mathds{1}_N \\
        -\mathds{1}_N & 0 
    \end{pmatrix} 
\end{equation}
denotes the standard symplectic form. Another important set of operators are the so-called displacement operators, whose effect is to modify the positions and momenta of the degrees of freedom by some specified amount. Specifically, for $\bxi \in \R^{2N}$ we have
\begin{equation}
     D_{\bxi} = D(\bxi) \coloneqq e^{-i\bxi^{T} \Omega \hat{\bold{x}}}
\end{equation}
and it follows from Eq.\eqref{eq:canonical_commutation_relations} that
\begin{equation}
    D_{\bxi}^{\dagger} \hat{\bold{x} }D_{\bxi} = \hat{\bold{x}} + \bxi \, .
\end{equation}
The displacement operators satisfy the commutation and multiplication relations
\begin{equation}
     D(\bxi)D(\bmu) =  e^{- \frac{i}{2}\bxi^{T} \Omega \bmu} D(\bxi + \bmu) = e^{-i \bxi^{T} \Omega \bmu} D(\bmu) D(\bxi) 
\end{equation}
and are mutually orthogonal with respect to the Hilbert-Schmidt inner product
\begin{equation}
     \tr \big(D_{\bxi}^{\dagger} D_{\bmu}\big) = (2\pi)^{N} \delta^{(2N)}(\bxi - \bmu) . 
\end{equation}
It is an essential fact that any (trace-class) operator $\hat{O}$ can be expanded in terms of displacements as
\begin{equation}
     \hat{O} = \frac{1}{(2\pi)^{N}} \int \diff\bxi\,  \chi_{\hat{O}}(\bxi) D_{\bxi} \, ,
\end{equation}
where $\chi_{\hat{O}}(\bxi) \coloneqq \tr (D^{\dagger}_{\bxi} \hat{O} )$ denotes the operator's characteristic function.

\section{Continuous-variable weight distributions}
\label{sec:macWilliams}
Let $\hat{O}_{1,2}$ denote two trace-class operators on the Hilbert space of $N$ modes. We define the following two  continuous weight distributions $\bA, \bB: \R_{\geq 0} \rightarrow \C$. First, the primary weight distribution
\begin{equation}
    \bA(r; \hat{O}_1, \hat{O}_2) = \int_{\norm{\bxi} = r} \diff \bxi \, \tr\big( D_{\bxi}^{\dagger} \hat{O}_{1}\big) \tr\big( D_{\bxi}^{\dagger} \hat{O}_{2} \big)^{*} \, ,
\end{equation}
and, second, the dual weight distribution
\begin{equation}
\label{eq:def:B_CV}
    \bB(r; \hat{O}_1, \hat{O}_2) = \int_{\norm{\bxi} = r} \diff \bxi \, \tr \Big( D_{\bxi} \hat{O}_{1} D_{\bxi}^{\dagger} \hat{O}_{2}^{\dagger} \Big)\, .
\end{equation}
Here both integrals are understood to be taken with respect to the surface measure $r^{2N-1} \diff \Omega^{(2N-1)}$. In the case of a single operator, i.e.\ $\hat{O}_{1} = \hat{O}_{2}$, both weight distributions become real-valued. The primary weight distribution $\bA$ then contains information about how much of an operator is supported on displacements of any fixed length, while the dual distribution $\bB$ contains information about the operator's overlap with copies shifted by some fixed length. 

We will from now on suppress the dependence on $\hat{O}_{1,2}$ in the arguments of $\bA$ and $\bB$. 

\subsection{The continuous-variable quantum MacWilliams identity}
Let us show that both distributions contain equivalent information content -- more precisely, they are related by an invertible linear integral transform. 
As both $\hat{O}_{1}$ and $\hat{O}_{2}$ are trace class we can expand them in terms of displacements to obtain the expression
\begin{equation}
    \tr \Big( D_{\bxi} \hat{O}_{1} D_{\bxi}^{\dagger} \hat{O}_{2}^{\dagger} \Big) = \frac{1}{(2\pi)^{2N}}\int \diff\boeta \diff\bmu \, \chi_{\hat{O}_{1}^{\phantom{\dagger}}}(\boeta) \chi_{\hat{O}_{2}^{\phantom{\dagger}}}(-\bmu)^{*} \tr \Big( D_{\bxi}  D_{\boeta}  D_{\bxi}^{\dagger}D_{\bmu}    \Big)\, .
\end{equation}
Upon evaluating the trace inside the integral to $(2\pi)^{N}\delta^{(2N)}(\boeta + \bmu) e^{-i \bxi^{T} \Omega \boeta}$ we obtain
\begin{equation}
    \tr \Big( D_{\bxi} \hat{O}_{1} D_{\bxi}^{\dagger} \hat{O}_{2}^{\dagger} \Big) = \frac{1}{(2\pi)^{N}}\int \diff\boeta \, \chi_{\hat{O}_{1}^{\phantom{\dagger}}}(\boeta) \chi_{\hat{O}_{2}^{\phantom{\dagger}}}(\boeta)^{*}\, e^{- i \bxi^{T} \Omega \boeta}\, . 
\end{equation}
Plugging this expression into the definition of the dual weight distribution~(\ref{eq:def:B_CV}) then yields
\begin{equation}
\label{eq:B_int_int}
    \bB(r; \hat{O}_1, \hat{O}_2) = \frac{1}{(2\pi)^{N}}\int \diff\boeta \, \chi_{\hat{O}_{1}^{\phantom{\dagger}}}(\boeta) \chi_{\hat{O}_{2}^{\phantom{\dagger}}}(\boeta)^{*}\, \int_{\norm{\bxi} =  r} \diff \bxi\,   e^{- i \bxi^{T} \Omega \boeta}\, .  
\end{equation}
We can identify the integral over $\norm{\bxi} = r$ as a Bessel function via the following integral expression, which we derive in Appendix~\ref{app:Bessel_identity}:
\begin{equation}
\label{eq:Bessel_zonal_spherical_function}
    J_{N-1}\big(r \norm{\boeta}\big) = \frac{\norm{\boeta}^{N-1}}{(2\pi r)^{N}} \int_{\norm{\bxi} = r} \hspace{-0.3cm}\diff\bxi \, e^{-i\bxi^{T} \Omega \boeta}\, .
\end{equation}
 Further splitting the integral over $\boeta$ in Eq.~\eqref{eq:B_int_int} into separate radial and angular components we then obtain the below relationship between the primary and dual weight distributions.
\begin{theorem}[Continuous-variable quantum MacWilliams identity] 
\label{thm:CVMW}
The primary weight distribution $\bA$ and dual weight distribution $\bB$ of any pair of trace class operators $\hat{O}_{1}$ and $\hat{O}_{2}$ are related by the identity
\begin{equation}
\label{eq:CVMW_BfromA}
    \bB(r; \hat{O}_1, \hat{O}_2) = r^{N} \int_{0}^{\infty} \diff x \, \frac{J_{N-1}(r x)}{x^{N-1}} \bA(x; \hat{O}_1, \hat{O}_2) \, .
\end{equation}
\end{theorem}
 By comparing with Eq.~\eqref{eq:radial_fourier_transform_background} we observe that this identity is exactly the Fourier transform for the radial part of the product of characteristic functions $\mathbf{r} \mapsto \chi_{\hat{O}_1^{\phantom{\dagger}}}(\mathbf{r}) \chi_{\hat{O}_2^{\phantom{\dagger}}}(\mathbf{r})^{*}$. This relation to the Fourier transform is a central aspect of the MacWilliams identities. It is present also in the classical MacWilliams identities, where the Krawtchouk matrices implement the radial Fourier transformation on functions on the Hamming graph~\cite{stanton_1990_Introduction, dunkl_1976_Krawtchouk}. The same applies to the discrete variable quantum MacWilliams identities of Shor and Laflamme. Here one can regard the Pauli group as a graph, with two Pauli operators connected by an edge if they differ in a single location. The resulting graph is again a Hamming graph and operators can be regarded as functions on this graph via expansion into Pauli operators. The radial Fourier transform on functions on this graph is then implemented by the Krawtchouk matrices and yields the discrete variable quantum MacWilliams identities of Shor and Laflamme when applied to the product of two functions obtained from operators.

It  follows from the fact that the Fourier transform is an involution that the integral kernel in Eq.~\eqref{eq:CVMW_BfromA} is also involutory and that we have the equivalent inverse relation
\begin{equation}
\label{eq:CVMW_AfromB}
    \bA(r) = r^{N} \int_{0}^{\infty} \diff x \, \frac{J_{N-1}(r x)}{x^{N-1}} \bB(x) \, .
\end{equation}
Last, we note that in the single-mode case ($N = 1$) this equation assumes the particularly simple form
\begin{equation}
\label{eq:CVMW_BfromA_single_mode}
  \bB(r) = r \int_{0}^{\infty} \diff x \, J_{0}(r x) \bA(x) \, .  
\end{equation}

\subsection{Examples}
Let us consider some examples of the continuous-variable quantum MacWilliams identity. In particular, we will consider examples where $\hat{O}_1 = \hat{O}_2= \Pi$ for the case of $\Pi$ the projector onto a coherent state, a Fock state, the codespace of a single mode cat code, and the codespace of a general GKP code.

\begin{figure}
    \centering
    \includegraphics[width=\linewidth]{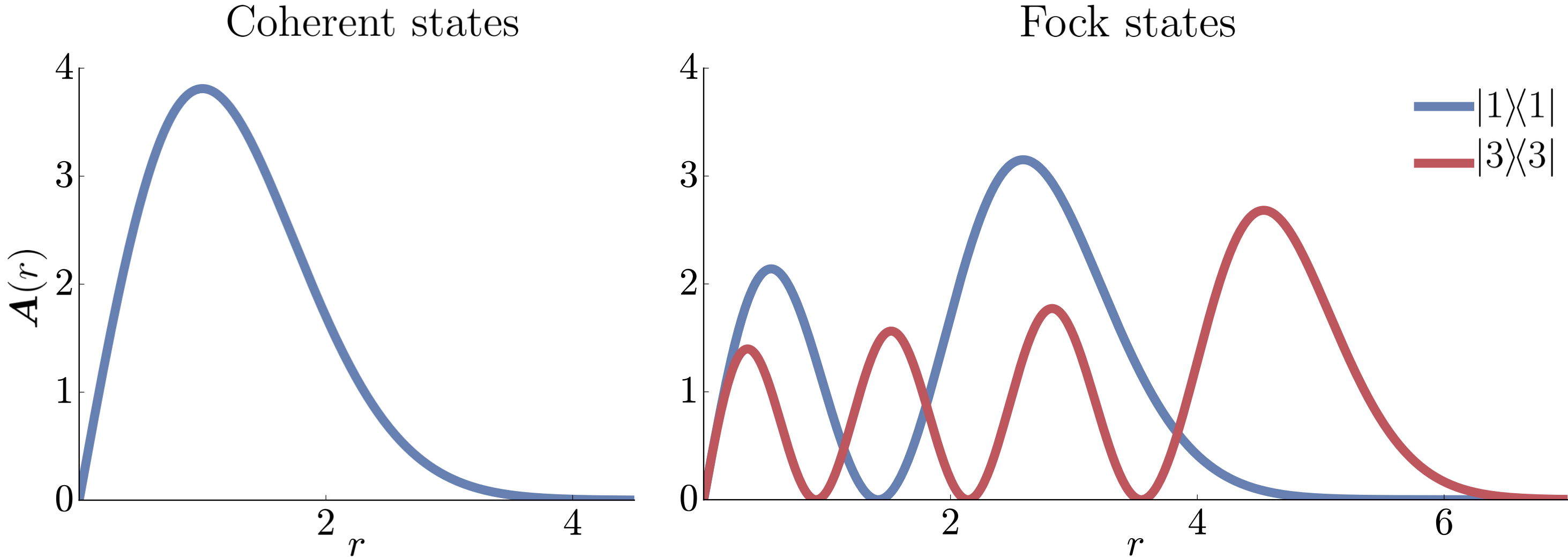}
    \caption{Weight distributions of example pure states. In the pure-state case the primary and the dual weight distributions coincide, i.e.\ $\bA = \bB$, and only the primary weight distribution is shown. Left: The weight distributions of a coherent states possess a single peak close to the origin, reflecting the localization of the state in phase space. The weight distribution is identical for all coherent states $\ket{\alpha}$. Right: The weight distributions of the first Fock state $\ket{1}$ and the third Fock state $\ket{3}$ are shown. In general, the weight distribution of the $n$-th Fock state oscillates $n+1$ times before decaying rapidly.}
    \label{fig:coherent_fock_AB}
\end{figure}

\subsubsection{Coherent states}
For a general coherent state projector on a single mode $\Pi_{\alpha} := \ket{\alpha}\bra{\alpha}$ with $\alpha \in \C$ the characteristic function can be evaluated directly and yields
\begin{equation}
\chi_{\Pi_{\alpha}^{\phantom{\dagger}}}(\bxi) \coloneqq \trb{D_{\bxi}^\dagger \ket{\alpha}\bra{\alpha}} = e^{i \boldsymbol{\alpha}^T \Omega \bxi} e^{- \frac{1}{4} \norm{\bxi}^2}\, ,
\end{equation}
where $\boldsymbol{\alpha} = \sqrt{2} (\re \alpha,\, \im \alpha)^T$ is a real vector associated to the complex number $\alpha$.
We immediately obtain the weight distribution 
\begin{equation}
    \bA(r) = 2 \pi r \Big|\chi_{\Pi_{\alpha}^{\phantom{\dagger}}}\big(\norm{\bxi} = r\big) \Big|^{2} = 2 \pi r e^{-\frac{1}{2} r^2}. 
\end{equation}
Direct evaluation further yields $\bB(r) = \bA(r)$. In fact, based on general considerations, we will see below that this is the case for all weight distributions of rank-1 orthogonal projectors. The continuous-variable quantum MacWilliams identities now reduce to the Bessel integral identity
\begin{equation}
  2 \pi r e^{-\frac{1}{2} r^2} = r \int_{0}^{\infty} \diff x \, J_{0}(r x) 2 \pi x e^{-\frac{1}{2} x^2} 
\end{equation}
which can be interpreted as an eigenvector equation for the employed Bessel integral kernel. 

\subsubsection{Fock states}
For the Fock-state projector $\Pi_{n} = \ket{n}\bra{n}$ we can again obtain the characteristic function as
\begin{equation}
    \chi_{\Pi_{n}^{\phantom{\dagger}}}(\bxi) = L_{n} \Big( \frac{1}{2} \norm{\bxi}^2\Big) e^{- \frac{1}{4} \norm{\bxi}^2}\, ,
\end{equation}
where $L_{n}$ denotes the $n$-th Laguerre Polynomial~\cite{barnett_2002_Methods}. The continuous-variable quantum MacWilliams identities in this case reduce to the integral identity
\begin{equation}
    2 \pi r L_{n} \Big( \frac{1}{2} r^2 \Big)^2 e^{- \frac{1}{2} r^2} = r \int_{0}^{\infty} \diff x \, J_{0}(r x) 2 \pi x L_{n} \Big( \frac{1}{2} x^2 \Big)^2 e^{- \frac{1}{2} x^2}\, .
\end{equation}
For a plot of the above weight distributions see Fig.~\ref{fig:coherent_fock_AB}.

\subsubsection{Cat codes}
\begin{figure}
    \centering
    \vspace{-1.5cm}
    \includegraphics[width=\linewidth]{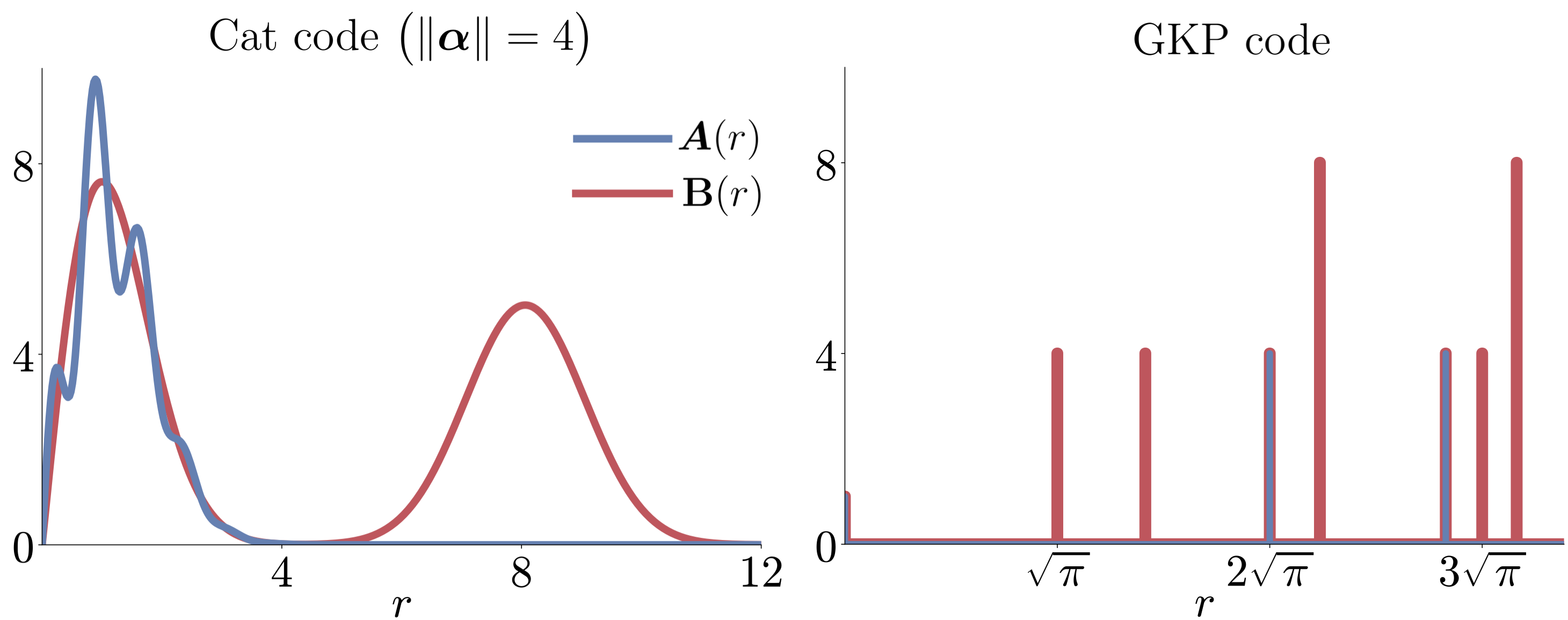}
    \caption{Weight distributions of example error correcting codes encoding a logical qubit. Left: Weight distributions of a cat code with well-separated coherent states ($\norm{\boldsymbol{\alpha}} = 4$). The distance between the coherent states is reflected in the dual distribution $\bB$, which has a second peak at $2 \norm{\boldsymbol{\alpha}}$, the length of a displacement required to shift a coherent state $\ket{\pm\alpha}$ onto the opposite coherent state $\ket{\mp\alpha}$. Right: Weight distributions of the single-mode square GKP code. Both primary and dual weight distributions are sums of integer multiples of delta functions, displayed as vertical lines. The primary weight distribution $\bA$ is the distribution of lengths of vectors within the lattice $\CL$ underlying the GKP stabilizer $\mathcal{S}$. For the square GKP code the this lattice is $\CL = 2\sqrt{\pi}\Z^{2}$ and contains a single vector of length $0$, as well as $4$ vectors of length $2\sqrt{\pi}$  and $\sqrt{8 \pi}$, among longer ones. This is reflected in the heights of peaks of the primary distribution $\bA$ at the respective lengths. The dual distribution $\bB$ contains equivalent information about the dual lattice; $\CL^{\perp} = \frac{1}{2}\CL$ in the case of a square GKP code with $K=2$. As GKP states are unnormalizable, both distributions are also unnormalizable and have peaks at arbitrarily large arguments $r$. The distance of the square GKP code is $\sqrt{\pi}$, indicated by the smallest argument $r$ where $\bB(r) > \bA(r).$}
    \label{fig:cat_GKP_AB}
\end{figure}

The single mode, two-legged cat code is a quantum error correcting code which protects against dephasing errors~\cite{albert_2022_Bosonic}. Its code space is the span of two coherent states $\mathcal{C} = \text{span}\big\{ \ket{\alpha}, \ket{-\alpha}\big\}$. For simplicity, we consider the case of well-separated coherent states, i.e.\ the case where the inner product $\langle \alpha | -\alpha \rangle = \text{exp}(- \norm{\boldsymbol{\alpha}}^2)$ is small. In this regime one can neglect the overlap between the well-separated coherent states and write the corresponding approximate cat code projector as $\Pi = \ketbra{\alpha} + \ketbra{-\alpha}$. It is straightforward to obtain the respective characteristic function 
\begin{equation}
    \chi_{\Pi}(\boldsymbol{\bxi}) = 4 e^{- \frac{1}{4} \norm{\bxi}^2} \cos\big(\boldsymbol{\alpha}^{T}\Omega\bxi\big)\, , 
\end{equation}
as well as the primary weight distribution
\begin{equation}
    \bA(r) = \int_{\norm{\bxi} = r} \diff \bxi \, |\chi_{\Pi}(\boldsymbol{\bxi})|^2 = 4 \pi r e^{-\frac{1}{2}r^2}\Big(1 + J_{0}(2 \norm{\boldsymbol{\alpha}}r)\Big)\, .
\end{equation}
The dual weight distribution can then be obtained either by direct computation, or through the MacWilliams identities~\eqref{eq:CVMW_BfromA} as 
\begin{equation}
\begin{split}
    \bB(r) & = 4 \pi r \int_{0}^{\infty} \diff x \, J_{0}(rx) x e^{-\frac{1}{2}x^2}\Big(1 + J_{0}(2 \norm{\boldsymbol{\alpha}}x)\Big)\\ &
    = 4 \pi r e^{-\frac{1}{2} r^2} \Big( 1 + e^{-2\norm{\boldsymbol{\alpha}}^2} I_{0}(2 \norm{\boldsymbol{\alpha}} r)\Big)\, ,
\end{split}
\end{equation}
where $I_{0}$ is the modified Bessel function of the first kind. The result follows from a known identity for integrals of Bessel functions~\cite{gradshteyn_2022_Table}. For a plot of these weight distributions see Fig~\ref{fig:cat_GKP_AB}.
\subsubsection{GKP codes}

Idealized GKP codes formally define code spaces containing infinite energy, unnormalizable codewords~\cite{gottesman_2001_Encoding}. As a consequence one cannot write down a well-defined trace-class projector associated to a GKP code, and these codes do not immediately fall within the framework developed above. However, due to their significance to quantum coding theory, we here give a separate derivation of the MacWilliams identities also for the case of ideal GKP codes. It is possible to derive expressions for weight distributions for GKP codes by associating to these codes formal projectors, such as in Eq.~\eqref{eq:gkp_ideal_projector}. However, we will take another approach, emphasizing the fact that GKP codes are stabilizer codes. Recall first that the stabilizer group of a GKP code consists of displacement operators. That is
\begin{equation}
    \mathcal{S} = \Big\langle e^{i \phi_1} D(\bxi_1), \dots, e^{i \phi_{2N}}D(\bxi_{2N}) \Big\rangle
\end{equation}
such that the vectors $\{\bxi_1, \dots, \bxi_{2N}\}$ generate a full rank lattice $\CL$, with symplectic inner products between lattice vectors in $2\pi \Z$~\cite{conrad_2022_GottesmanKitaevPreskill, burchards_2024_Fiber}. The normalizer of the GKP stabilizer within the displacement operators modulo phases is then simply given by the set of logical displacement Pauli operators 
\begin{equation}
    \mathcal{N(S)} = \Big\{ D(\bxi) \Bigmid \bxi \in \CLp \Big\}\, .
\end{equation}
Here $\CLp \subseteq \CL$ is the symplectic dual lattice of $\CL$, defined by 
\begin{equation}
    \CLp = \Big\{ \bxi^{\perp} \in \R^{2n} \Bigmid \bxi^{T} \Omega \bxi^{\perp} \in 2 \pi \Z \, \text{ for all } \bxi \in \CL  \Big\}\, .
\end{equation}
Let us recall now that the weight distributions of qudit stabilizer codes possess a particularly simple combinatorial interpretation. In particular, for a qudit stabilizer code of dimension $K$, the primary distribution $\bA_{i}/K^2$ simply counts the number of stabilizer elements of a certain Pauli weight, while the dual weight distribution $\bB_{i}/K$ counts the number of normalizer elements of a certain Pauli weight~\cite{gottesman_1997_Stabilizer}. In the continuous-variable case the severity of a displacement operator is given by the norm of its argument, a real number. The appropriate generalization to the continuous-variable case is then to introduce two distributions, 
$\bA$ and $\bB$, whose integrals over a given interval yield the number of stabilizer or normalizer elements with weights falling within that interval. Concretely, let us denote by $\CD$ the set of distances assumed by vectors within $\CL$ 
\begin{equation}
    \CD = \Big\{d\in \R \Bigmid  \exists \bxi \in \CL: \norm{\bxi} = d\Big\}
\end{equation}
and by $N_{d}$ their respective multiplicities
\begin{equation}
    N_{d} = \Big|\Big\{\bxi \in \CL \Bigmid  \norm{\bxi} = d\Big\}\Big|\, ,
\end{equation}
as well as by $D^{\perp}$ and $N_{d}^{\perp}$ the corresponding quantities for $\CLp$. Then we have 
\begin{equation}
\label{eq:A_ideal_gkp}
    \bA(r) = K^{2} \sum_{d \in \CD} N_{d} \delta(r - d) 
\end{equation}
and 
\begin{equation}
\label{eq:B_ideal_gkp}
    \bB(r) = K \sum_{d\in \CD^{\perp}} N_{d}^{\perp} \delta(r - d) \, .
\end{equation}
Let us now show explicitly that the two weight distributions~\eqref{eq:A_ideal_gkp} and~\eqref{eq:B_ideal_gkp} are related by the MacWilliams identities. Let's consider the case of the dual distribution $\bB$ acting on a rapidly decaying function $f$. We have 
\begin{equation}
    \begin{split}
     \int_{0}^{\infty} \diff r\,  \bB(r) f(r) 
     &= K \int_{0}^{\infty} \diff r\,  \sum_{d\in \CD^{\perp}} N_{d}^{\perp} \delta(r - d) f(r) \\
     &= K \int_{0}^{\infty} \diff r\,  \sum_{\bxi^{\perp} \in \CLp } \delta\big(r - \norm{\bxi}\big) f(r) \\
    &= K \sum_{\bxi^{\perp} \in \CLp } f(\norm{\bxi}) \, .
    \end{split}
\end{equation}
We now emply the Poisson summation formula which states that for any sufficiently rapidly decaying function $\sum_{\bxi^{\perp} \in \CLp } \hat{f}(\bxi^{\perp}) = |\text{det } \CL/\sqrt{2\pi} |^{\frac{1}{2}} 
 \sum_{\bxi \in \CL} f(\bxi) $~\cite{elkies__Theta}. Together with Eq.~\eqref{eq:radial_fourier_transform_background} for the Fourier transform of a radial function, and  the fact that for a GKP codes lattice the relationship $K^2 = |\text{det } \CL/\sqrt{2\pi} |$ connects code size and lattice determinant \cite{conrad_2022_GottesmanKitaevPreskill}, we can then rewrite the above expression as
\begin{equation}
\begin{split}
& = K^{2} \sum_{\bxi \in \CL } \hat{f}(\norm{\bxi}) \\
& = K^{2} \sum_{\bxi \in \CL }\frac{1}{\norm{\bxi}^{N-1}} \int_{0}^{\infty} \diff r \, r^{N}J_{N-1}(r \norm{\bxi}) f(r) \\
& = K^{2} \int_{0}^{\infty} \diff r \, r^{N} \int_{0}^{\infty} \diff x \, \frac{J_{N-1}(r x)}{x^{N-1}} \sum_{\bxi \in \CL } \delta\big(x - \norm{\bxi}\big) f(r) \\
    & = \int_{0}^{\infty} \diff r \, \Bigg(r^{N} \int_{0}^{\infty} \diff x \, \frac{J_{N-1}(r x)}{x^{N-1}} \bA(x) \Bigg) f(r) \, .
\end{split}
\end{equation}
As two distributions are equal if they have identical action on functions, we observe the the quantity in parenthesis must be equal to $\bB$. Comparing with Eq.~\eqref{eq:CVMW_BfromA} this is exactly the continuous-variable quantum MacWilliams identity previously derived for trace-class operators. 

The above MacWilliams identity for GKP codes is a consequence of the general fact that the length distribution of a lattice determines that of its dual lattice, sometimes stated in the language of theta functions~\cite{elkies__Theta}. Such theta functions have been previously employed to analyze concatenated stabilizer-GKP codes, as well as the decoding problem in GKP codes~\cite{conrad_2022_GottesmanKitaevPreskill, conrad_2024_Good}.

\subsection{Properties of weight distributions}
Let us investigate some of the properties of the distributions $\bA$ and $\bB$. First, we will derive some general properties which hold for any operator pair $\hat{O}_1, \hat{O}_2$. We then go on to derive some additional properties for the special case where $\hat{O}_1 = \hat{O}_2 = \Pi$ is the codespace projector of a QECC of distance $d$.
\subsubsection{Invariance properties}
\label{sec:AB_euclidean_invariance_properties}
The distributions $\bA, \bB$ are invariant under the simultaneous transformation 
\begin{equation}
    \hat{O}_{1,2} \mapsto D(\bxi) \hat{O}_{1,2} D(\bxi)^{\dagger} 
\end{equation}
as well as under passive linear transformations, i.e.\ those Gaussian unitary operators $U$ that act on displacement operators as $U D(\bxi) U^{\dagger} = D(Q \bxi)$ with $Q$ an orthogonal map. Together this means that both weight distributions are invariant under the motion group of $2N$-space, the Euclidean group $E(2N)$. The analogous property of Shor-Laflamme enumerators is their invariance under products of arbitrary transversal unitaries and $\swap$ gates~\cite{shor_1997_Quantum}.

\subsubsection{Normalization of weight distributions}
Using the property of the 2-mode $\swap$ gate that $\tr \swap (\hat{O}_1 \otimes \hat{O}_2) = \tr \hat{O}_1 \hat{O}_2$ for any pair of trace-class operators $\hat{O}_1$ and $\hat{O}_2$, together with its decomposition into displacement operators, $\swap = \frac{1}{(2\pi)^{N}} \int \diff \bxi \, D_{\bxi}^{\dagger} \otimes D_{\bxi}$, we can show that
\begin{equation}
\begin{split}
\label{eqs:A_sum}
     \int_{0}^{\infty} \diff r\, \bA(r)& =  \int \diff \bxi \, \trb{(D_{\bxi}^{\dagger} \otimes D_{\bxi}) (\hat{O}_{1} \otimes \hat{O}_{2}^{\dagger})}\\ & =  \, (2 \pi)^{N} \trb{\swap (\hat{O}_{1} \otimes \hat{O}_{2}^{\dagger})} = (2 \pi)^{N} \trb{ \hat{O}_{1} \hat{O}_{2}^{\dagger}} \, .
\end{split}
\end{equation}
In a similar way we show in Appendix~\ref{app:Integral_B} that
\begin{equation}
\label{eqs:B_sum}
    \int_{0}^{\infty} \diff r \,  \bB(r)  = \int \diff \bxi\, \tr \big( D_{\bxi} \hat{O}_{1} D_{\bxi}^{\dagger} \hat{O}_{2} \big)  
     = (2 \pi)^{N} \trb{\hat{O}_{1}} \trb{\hat{O}_{2}^{\dagger}}\, .
\end{equation}
In this sense $\bA$ and $\bB$ are radial decompositions of the respective operator traces. Analogous properties of Shor-Laflamme enumerators have been derived in a similar manner in Ref.~\cite{huber_2018_Bounds}. 

\subsubsection{Weight distributions of error-correcting codes}
Consider now the weight distributions $\bA$ and $\bB$ for the special case where $\hat{O}_1 = \hat{O}_2 = \Pi$ with $\Pi$ a distance-$d$ quantum error correcting code projector encoding a logical subspace of dimension $K = \tr \Pi$. The quantum error correcting conditions, established by Knill and Laflamme, state that the projector satisfies 
\begin{equation}
\label{eq:knill_laflamme_conditions}
    \Pi D(\bxi) \Pi = c(\bxi) \Pi \quad \text{for all}\quad \norm{\bxi} < d 
\end{equation}
where $c$ is a scalar function~\cite{knill_2000_Theory}. We may follow Shor and Laflamme~\cite{shor_1997_Quantum} and write the expressions for both weight distributions by using an orthonormal basis $\ket{i}$ for the codespace
\begin{equation}
\begin{split}
\label{eq:AB_ONB_decomposition}
    \bA(r) & = \int_{\norm{\bxi} = r}\diff \bxi\, \Big| \sum_{i} \bra{i} D(\bxi) \ket{i}\Big|^{2} \\
    \bB(r) & = \int_{\norm{\bxi} = r} \diff \bxi\, \sum_{i,j}\big| \bra{i} D(\bxi) \ket{j}\big|^{2} \, .
\end{split}
\end{equation}
It follows immediately that $0 \leq \bA, \bB  $. The Cauchy-Schwarz inequality for the Hilbert-Schmidt inner product between the codespace projector $\Pi$ and the operator $\Pi D(\bxi) \Pi$ further shows that $\bA \leq K \bB$. In particular, for pure state projectors $\Pi = \ket{\psi}\! \bra{\psi}$ this inequality together with Eqs.~(\ref{eqs:A_sum}, \ref{eqs:B_sum}) implies that $\bA = \bB$. For a QECC projector we obtain from the Knill-Laflamme conditions that for $\norm{\bxi} < d \colon$ 
\begin{equation}
    \trb{D_{\bxi} \Pi} \trb{D_{\bxi}^{\dagger} \Pi } = K c(\bxi) \trb{D_{\bxi}^{\dagger} \Pi} = K \trb{D_{\bxi} \Pi D_{\bxi}^{\dagger}\Pi }
\end{equation}
and thus $\forall r < d\colon \bA(r) = K \bB(r)$, saturating the previous inequality for all arguments below the distance. 

\subsection{Weight distributions in terms of fidelities}
Both weight distributions of a projector $\Pi$ have physical interpretations in terms of fidelities. One such fidelity, the entanglement fidelity, is a measure of how well a state and its entanglement with any outside degrees of freedom are preserved under a noise channel~\cite{nielsen_1996_entanglement}. A general result from Ref.~\cite{schumacher_1996_Sending} on the entanglement fidelity states that given any quantum channel $\CE$ with Kraus operators $K_{i}$ the entanglement fidelity of a state w.r.t.\ the channel can be expressed as

\begin{equation}
    F_{e}(\rho; \CE) = \sum_{i} \trb{K_{i} \rho} \trb{K_{i}^{\dagger} \rho } \, .
\end{equation}
 Let us introduce the uniformly random distance-$r$ displacement channel on $N$ modes
\begin{equation}
    \CN_{r}(\rho) \coloneqq \frac{1}{S^{(2N-1)}(r)} \int_{\norm{\bxi} = r}\diff \bxi\, D_{\bxi} \rho D_{\bxi}^{\dagger}\, ,
\end{equation}
where $S^{(n)}(r)$ denotes the surface area of the $n$-sphere of radius $r$. We see that the primary weight distribution can be expressed in terms of the entanglement fidelity of the maximally mixed logical state $\overline{\rho} = \Pi/K$ w.r.t.\ the displacement channel
\begin{equation}
    \frac{1}{K^{2} S^{(2N-1)}(r)} \bA(r) = F_{e}(\overline{\rho}, \CN_{r}) \, .
\end{equation}
Similarly, the dual weight distribution can be expressed using the displacement channel as follows
\begin{equation}
\label{eq:B_tr_rho_Nrho}
     \frac{1}{K^{2} S^{(2N-1)}(r)} \bB(r) = \trb{\overline{\rho} \CN_{r}(\overline{\rho})} \, ,
\end{equation}
and this expression can be interpreted as the probability of a code state remaining within the code space upon traversing $\CN_{r}$.

\section{Lower bound on average occupation number from weight distributions}
\label{sec:lower_bound_from_B}
The weight distributions $\bA$ and $\bB$ of a projector $\Pi$ not only contain information about its trace and its error correction properties, but also about the expected occupation number $\overline{n} = \frac{1}{K}\tr \hat{n} \Pi$ of a maximally mixed logical state. In particular, the second moment of the $\bB$ distribution provides a lower bound on the expected occupation number. In order to show this we will make use of two phase-space representations of a quantum state, the Husimi $Q$-function as well as the Glauber-Sudarshan $P$ representation. Throughout this section $\bA, \bB$ will refer to the distributions of the maximally mixed logical state $\rho = \frac{1}{K} \Pi$. For simplicity, we will only provide derivations for the single mode ($N = 1$) setting in this section.

\subsection{Properties of phase-space distributions}
Let us recall that the Husimi $Q$-function $Q(\alpha)$ and the Glauber-Sudarshan $P$ representation $P(\alpha)$ of a quantum state $\rho$ are defined in terms of coherent states as 
\begin{equation}
    Q(\alpha) \coloneqq \frac{1}{\pi} \bra{\alpha} \rho \ket{\alpha} \geq 0
\end{equation}
and
\begin{equation}
    \rho = \int_{\C} \diff \alpha\,  P(\alpha) \ket{\alpha} \bra{\alpha} \, .
\end{equation} 
We denote by $\varphi(\alpha) = \frac{1}{\pi}|\langle 0 | \alpha \rangle |^{2} = \frac{1}{\pi} e^{-|\alpha|^{2}}$ the Gaussian function and note that the $P$ and $Q$ functions are connected via the relation 
\begin{equation}
    Q(\alpha) = (\varphi \ast P)(\alpha) \coloneqq \int_{\C} \diff\beta\,  \varphi(\alpha - \beta) P(\beta)\, ,
\end{equation}
where $\ast$ denotes the convolution operation. 
The phase-space distributions $P$ and $Q$ are those corresponding to normal and antinormal operator ordering respectively. Hence, computing moments under these distributions yields expectation values of correspondingly ordered operators.

In particular, denoting by $\rho_{\bxi} = D_{\bxi} \rho D_{\bxi}^\dagger$ the displaced state, its expected occupation number can be expressed in terms of the $P$ distribution as follows:
\begin{equation}
\label{eq:P_shifted_occ_number}
    \tr (\hat{n} \rho_{\bxi}) = \int_{\C} \diff \alpha\,  P(\alpha)\tr \big(\hat{n} D_{\bxi} \ket{\alpha} \bra{\alpha}D_{\bxi}^\dagger\big)
    = \int_{\C} \diff \alpha\,  P(\alpha) |\alpha+\tilde{\xi}|^{2} \, ,
\end{equation}
where we have denoted by $\tilde{\xi} = \frac{1}{\sqrt{2}}(\bxi_1 + i \bxi_2)$ the complex number associated to the vector $\bxi \in \R^2$. The corresponding antinormal ordered expectation value is a moment of the $Q$-function
\begin{equation}
\label{eq:Q_shifted_occ_number}
    \trb{\hat{n} \rho_{\bxi}} +1 =\trb{a a^\dagger \rho_{\bxi}} = \int_{\C} \diff \alpha\,  Q(\alpha) |\alpha+\tilde{\xi} |^{2} \, ,
\end{equation}
where $a$ and $a^{\dagger}$ denote the creation and annihilation operators.

\subsection{Dual weight distribution in terms of the $P$ and $Q$ functions}
It is straightforward to obtain an expression for the dual weight distribution $\bB$ in terms of the $P$-function. Combining the definition of the dual distribution with that of the $P$-function we find
\begin{equation}
\begin{split}
    \bB(r) & = \int\limits_{\norm{\bxi}=r}\diff\bxi \int_{\C} \diff\alpha \diff\beta\, P(\alpha) P(\beta) \tr \Big( D_{\bxi}^{\dagger} \ketbra{\alpha} D_{\bxi} \ketbra{\beta} \Big) \\ &
    = \pi \int\limits_{\norm{\bxi}=r}\diff\bxi \int_{\C} \diff\alpha \diff\beta\, P(\alpha) P(\beta) \varphi\Big(\alpha - \beta -  \tilde{\xi}\Big) \\ &
     = \pi \int\limits_{\norm{\bxi}=r}\diff\bxi \int_{\C} \diff\alpha \diff\beta \,P(\alpha) P(\beta)  \Big( \varphi(x) \ast_{x} \delta \big(x - \beta -  \tilde{\xi}\big)\Big)(\alpha)\\ &
    = \pi \int\limits_{\norm{\bxi}=r}\diff\bxi \int_{\C} \diff\alpha \diff\beta \,Q(\alpha) P(\beta) \delta\Big(\alpha - \beta -  \tilde{\xi}\Big)\, .
\end{split}
\end{equation}
We can now obtain an expression for the second moment of the $\bB$ distribution purely in terms of the $Q$-function
\begin{equation}
\begin{split}
\label{eq:B_second_moment_PQ_QQ}
    \int_{0}^{\infty} \diff r\, \bB(r) r^{2} & = \pi \int \diff\bxi\, \norm{\bxi}^{2} \int_{\C} \diff \alpha \diff\beta\,  Q(\alpha) P(\beta) \delta\Big(\alpha - \beta -  \tilde{\xi}\Big) \\ & 
    = 4 \pi \int_{\C} \diff\alpha \diff\beta\, Q(\alpha) P(\beta) |\alpha-\beta|^{2} \\ & 
    = 4 \pi \int_{\C}  Q(\alpha) Q(\beta) |\alpha-\beta|^{2} \diff\alpha \diff\beta - 4 \pi 
    \, .
\end{split}
\end{equation}
Here we have made use of Eqs.~(\ref{eq:P_shifted_occ_number}) and~(\ref{eq:Q_shifted_occ_number}) in the last step.
This expression for the second moment of $\bB$ purely in terms of the Husimi $Q$-function will be our main tool in deriving a lower bound on the expected occupation number of the state $\rho$.
First, however, note that we can also do the opposite and lower bound moments of $\bB$ by the occupation number expectation value minimized over shifted versions of the state $\rho$. Specifically, we have, from equations~(\ref{eq:B_second_moment_PQ_QQ}) and~(\ref{eq:P_shifted_occ_number}), that
\begin{equation}
\begin{split}
\label{eq:B_second_moment_lower_bound}
    \min_{\bxi}\trb{\hat{n} \rho_{\bxi}}  \leq 
     \int_{\C} \diff\alpha\,  Q(\alpha) \trb{\hat{n} \rho_{\boldsymbol{\alpha}}}
     = \frac{1}{4 \pi}\int_{0}^{\infty} \diff r\, \bB(r) r^{2} 
\end{split}
\end{equation}
with $\boldsymbol{\alpha} = \sqrt{2}(\re \alpha, \, \im \alpha)$.
To obtain an upper bound on the second moment in terms of occupation numbers let us first note that for any pair of complex numbers one has that $|\alpha-\beta|^2 \leq 2 (|\alpha|^2 + |\beta|^2)$. This inequality together with Eq.~(\ref{eq:Q_shifted_occ_number}) and normalization of the $Q$-function then implies that
\begin{equation}
\begin{split}
    \int_{\C} \diff\alpha \diff\beta \,  Q(\alpha) Q(\beta) |\alpha-\beta|^{2} & 
    \leq 2\int_{\C} \diff\alpha \diff\beta \,  Q(\alpha) Q(\beta) \big(|\alpha|^{2} + |\beta|^{2} \big) \\ &
    = 4(\overline{n} + 1) \, .
\end{split}
\end{equation}
Inserting into Eq.~\eqref{eq:B_second_moment_PQ_QQ} provides us with the following upper bound on the second moment in terms of occupation numbers
\begin{equation}
    \int_{0}^{\infty} \diff r\, \bB(r) r^{2} \leq 4 \pi (4 \overline{n} + 3) \, .
\end{equation}
Suitably generalized to the $N$-mode setting the above calculation yields multi-mode versions of this upper bound and the lower bound~\eqref{eq:B_second_moment_lower_bound} which read
\begin{equation}
\label{eq:B_multimode_occ_doublesided_inequality}
    \min_{\bxi} \tr \hat{n} \rho_{\bxi} \leq \frac{1}{4 \pi} \int_{0}^{\infty} \diff r\, \bB(r) r^{2}\leq (4 \overline{n} + 3N) \, ,
\end{equation}
with the total occupation number operator $\hat{n} = \sum_{i} \hat{n}_{i}$. From an intuitive point of view the inequalities~\eqref{eq:B_multimode_occ_doublesided_inequality} encode a relationship between the extent of a state $\rho$ in phase-space and its occupation number. As a consequence of the fact that it is preserved under Euclidean operations on phase-space, as discussed in Section~\ref{sec:AB_euclidean_invariance_properties}, the $\bB$ distribution must be ignorant to the exact occupation number of a state. However, it encodes information about the total extent of the characteristic function, and hence the $Q$-function, in phase-space. As occupation number is a quadratic function on phase-space, information about the concentration of the non-negative $Q$-function then yields a lower bound on the occupation number.

\section{Error detection}
\label{sec:approximate}
Above we have shown that the distances of ideal QECCs in the sense of Knill and Laflamme are reflected in the weight distributions of codes. Let us now discuss the problem of error detection, and a relaxation of this problem away from the ideal case. 
\subsection{Exact error detection}
Consider the following problem: An unknown state $\rho$ is encoded into a subspace $\CC$ with associated projector $\Pi$. An adversary is then given the choice to either apply a quantum channel $\CN(\cdot) = \sum_{i} E_{i} \cdot E_{i}^\dagger$ or to leave the state unchanged. We would like to perform a measurement which detects whether the state has been modified with maximal probability and which does not disturb the original state if unmodified. The optimal such probability is achieved by a measurement consisting of the code space projector and its orthogonal complement $\{\Pi, \Pi^\perp\}$. Under what circumstances can we guarantee recovery of the original state upon not detecting an error? This is the case if states $\CN(\rho)$ are guaranteed to collapse onto $\rho$ when measured to lie within $\Pi$, i.e.\ whenever $\forall i\colon \Pi E_i \Pi \propto \Pi$. These are exactly the quantum error correction conditions we have discussed before.

\subsection{Approximate error detection}
Let us now consider the following approximate version of the above error detection
problem: A maximally entangled state $\ket{\psi}$ is created between the code space $\CC$ on system $Q$ and a reference
system $R$ of dimension $K = \textnormal{dim}(\CC)$. An adversary is given the choice to either apply a quantum channel $\CN$ to the code system $Q$ or to leave the state unchanged. Again we would like to perform a measurement which detects whether the state has been modified with maximal probability and which does not disturb the original state if unmodified. As discussed above, the detection probability is maximized by the measurement $\{\ketbra{\psi}, \mathds{1}-\ketbra{\psi} \}$. However, suppose we do not have
access to the reference system $R$ and are restricted to performing measurements on the code system $Q$ only. The optimal error detection probability is then again achieved by the measurement $\{\Pi, \Pi^\perp\}$. What is the probability $p_{\err}$ that the original entangled state $\ketbra{\psi}$ has changed in spite of no error having been detected following application of the noise channel $\CN$? It can be shown that $p_{\err} = 0$ exactly if $\CC$ is a QECC, i.e.\ satisfies the Knill-Laflamme conditions. Let us then define a quality-$\epsilon$ quantum error detecting code (QEDC) to be one that achieves $p_{\err} \leq \epsilon$, that is it recovers the entangled state $\ketbra{\psi}$ in fraction $1-\epsilon$ of cases where noise has been applied and no error detected.
\begin{definition}[Approximate quantum error detection code]
    A $K$-dimensional subspace $\CC$ of the Hilbert space of $N$ modes is an $[[N, K, d, \epsilon]]$-QEDC if it achieves $p_{\err} \leq \epsilon$ for arbitrary convex combinations of the channels $\CN_{r}$ with $ r < d$.
\end{definition}
Let us emphasize again that for $\epsilon = 0$ this definition corresponds to an ideal error correcting code of distance $d$. We can now show that the above notion is directly related to the weight distributions of the code. 
\begin{lemma}
    Let $\bA$ and $\bB$ be the weight distributions of the projector $\Pi$ onto a subspace $\CC$ of dimension $K$ on $N$ modes. Then $\CC$ is an $[[N, K, d, \epsilon]]$-QEDC if and only if
    \begin{equation}
        \label{eq:A_B_frac_exceeds_epsilon}
        \frac{\bA(r)}{K\bB(r)} \geq 1- \epsilon
    \end{equation}
    for all $r < d$.
\end{lemma}
\begin{proof}
    Consider the channel $\CN = \int_{0}^{d} p(r) \CN_{r} \diff r$ for a probability distribution $p$  with support on $[0, d)$. Let further $\overline{\rho} = \tr_{R} \ketbra{\psi}$ be the maximally mixed state on $\CC$. Then the probability of detecting no error via measurement of the codespace projector is given by $p_{1} = \tr \Pi \CN(\overline{\rho}) = \int_{0}^{d} \frac{p(r) \bB(r)}{K S^{(2N-1)}(r)}\,  \diff r\,$. The probability of no error occurring on the global entangled state is $p_{2} = \bra{\psi} (\mathds{1} \otimes \CN)(\ketbra{\psi}) \ket{\psi} = \int_{0}^{d} \frac{p(r) \bA(r)}{K^{2} S^{(2N-1)}(r)}\, \diff r .$ Out of all cases where no error has been detected, the fraction where the global state has indeed changed is hence 
    \begin{equation}
        p_{\err} = \frac{p_1 - p_2}{p_{1}}= 1 - \frac{\int_{0}^{d} p(r) \bA(r)/ S^{(2N-1)}(r) \diff r}{\int_{0}^{d} p(r) K \bB(r)/ S^{(2N-1)}(r) \diff r} \, .
    \end{equation}
    This term is upper bounded by, and can approach arbitrarily closely, $\textnormal{sup}\{ 1 - \bA(r)/K\bB(r) \mid {r\in [0, d)} \}$
    which does not exceed $\epsilon$ for any distributions $p$ if and only if Condition~(\ref{eq:A_B_frac_exceeds_epsilon}) is satisfied.
\end{proof}

We use this relationship to restrict the existence of possible QEDCs in the next section. In Appendix~\ref{app:finite_energy_GKP} we estimate the code parameters of an approximate single mode GKP code, produced from its idealized infinite energy limit via an envelope operator $E_\Delta = \text{exp}(- \Delta^2 \hat{n})$. We show that an ideal $[[N=1, K, d]]$-GKP code produces approximate QEDCs of parameters $[[N, K, d/2-\delta, \epsilon]]$ with $\epsilon \in \mathcal{O}(\text{exp}(-\frac{d \delta}{8 \Delta^2} ))$ for all $\delta > 0$. This presents a discontinuous drop in distance, as introduced above, at $\epsilon=0$. The reason for this drop is that when employing ideal GKP codes for the above detection task, displacement errors of length below the distance will be detected whenever they occur. For an approximate code, however, if a logical $\ket{\overline{0}}$ state is displaced by more than half of the ideal 
 code's distance, it is closer in phase-space to a logical $\ket{\overline{1}}$ state for some directions. Post-selecting on the outcome where no error was detected, it will then collapse onto $\ket{\overline{1}}$ with probability approaching unity as $\Delta \rightarrow 0$.

\section{Bounds from weight distributions}
\label{sec:programming_bounds}
Let us now derive a bound on the parameters of an $[[N, K, d, \epsilon]]$-QEDC with weight distributions $\bA, \bB$. Suppose $\hat{f}(x) \geq 0$ is a nonzero function on $\R_{\geq 0}$ such that its radial Fourier transform
\begin{equation}
    \label{eq:radial_fourier_trafo_f}
    f(y) = \frac{1}{y^{N-1}}\int_{0}^{\infty} \diff x \, J_{N-1}(xy) x^{N} \hat{f}(x) 
\end{equation}
satisfies $f(y) \geq 0$ for all $y < d$ and $f(y) \leq 0$ for $y \geq d$. Then, we have that
\begin{equation}
\begin{split}
      \int_{0}^{d} \diff x\,  \hat{f}(x) \bA(x)  & 
     \leq \int_{0}^{\infty}\diff x\,  \hat{f}(x) \bA(x) \\ &
     = \int_{0}^{\infty}\diff x\,  \hat{f}(x) x^{N} \int_{0}^{\infty} \diff y\, \frac{J_{N-1}(xy)}{y^{N-1}} \bB(y)\\ &
     = \int_{0}^{\infty} \diff y\, f(y) \bB(y)\\ &
     \leq \int_{0}^{d} \diff y\, f( y) \bB(y) \\ & \leq \frac{1}{K (1 - \epsilon)} \int_{0}^{d}\diff y\,  f(y) \bA(y) \, .
\end{split}
\end{equation}
Here we have employed, in order, Eqs.~(\ref{eq:CVMW_BfromA}), 
(\ref{eq:radial_fourier_trafo_f}) and (\ref{eq:A_B_frac_exceeds_epsilon}). Let us note that, in order to ensure convergence for all of the above expressions, it suffices to assume $f$ and $\hat{f}$ bounded. Rearranging, it follows that 
\begin{equation}
    K \leq \frac{1}{1-\epsilon} \frac{\int_{0}^{d}\diff x\,  f(x) \bA(x) }{\int_{0}^{d} \diff x\,  \hat{f}(x) \bA(x)} \leq \frac{1}{1-\epsilon} \sup\Bigg\{\frac{f(x)}{\hat{f}(x)} \Biggmid x \in [0, d] \Bigg\} \, ,
\end{equation}
and we state this result as our main theorem.

\QuantumCohnElkies*
This quantum version of the Cohn-Elkies bound differs from the classical version mainly through the occurrence of a supremum, an intrinsic quantum feature which occurs also in the discrete variable version of the bound~\cite{ashikhmin_1999_Upper}. 

\subsection{The Levenshtein bound}
We will now derive a version of Levenshtein's classical sphere packing bound \cite{levenshtein_1979_bounds} for quantum codes. The classical bound is the second-best asymptotic bound on the density of sphere packings known, the tightest bound in sufficiently high dimensions being a constant factor improvement on that of Kabatiansky and Levenshtein \cite{kabatiansky_1978_Bounds, cohn_2014_Sphere}. Our approach is based on the application of the above general bound to the Levenshtein function $f: \R^{2N} \rightarrow \R$ given by
\begin{equation}
\label{eq:levenshtein_f_from_ghat}
    f(\mathbf{x}) = \big(1-x^2\big) \hat{g}\big( \mathbf{x} \big)^2
\end{equation}
with 
\begin{equation}
\label{eq:g_definition}
    \hat{g}(\mathbf{x}) = \frac{N!}{1-x^2} \frac{2^{N}}{ j_{N}^{N} x^{N}}J_{N}(j_{N} x) \, .
\end{equation}
Here $j_{N}$ denotes the first positive zero of the Bessel function $J_{N}$.
The function $f$ can be obtained by variational considerations, optimizing the classical Cohn-Elkies bound over functions whose Fourier transforms have compact support. It can be employed to prove the classical Levenshtein sphere packing bound via the Cohn-Elkies bound~\cite{cohn_2003_New}. The functions are normalized and scaled such that $f(0) = \hat{g}(0) = 1$ and such that $f$ achieves its first zero at $x = 1.$ The Fourier transform $g$ is available explicitly  and is given by 
\begin{equation}
    g(\mathbf{x}) = c_{N} \Bigg( 1 - \frac{j_{N}^{N-1}}{J_{N-1}(j_{N})} \frac{J_{N-1}(x)}{x^{N-1}}\Bigg) \chi_{j_{N}} (x) \, ,
\end{equation}
where $c_N = 2^{N} N!/j_{N}^{2N}$ and where $\chi_r$ denotes the indicator function of the radius-$r$ Ball. It is shown in Ref.~\cite{cohn_2003_New} that the Fourier transform of $f$ is given by a convolution of $g$ with the indicator function of a Ball
\begin{equation}
    \hat{f}(\mathbf{x}) = c_{N} (g * \chi_{j_{N}})(\mathbf{x}) \, .
\end{equation}
As $g$ and $\chi_{j_{N}}$ are everywhere non-negative the same holds for $\hat{f}$. Together with Eq.~\eqref{eq:levenshtein_f_from_ghat} this shows that $f$ satisfies the conditions of Theorem~\ref{thm:quantum_cohn_elkies}.

In Appendix~\ref{app:levenshtein_d_proof} we will be occupied with the proof of the following technical Lemma concerning the Levenshtein function.
\begin{lemma}
    \label{lem:levenshtein_supremum}
    For $0 < d \leq d_{+} = \Big( \frac{16 N! |J_{N-1}(j_N)| }{3\sqrt{\pi} \Gamma\big(\frac{2N-1}{2} \big) j_{N}^{N-2}} \Big)^{1/6}$ the Levenshtein function satisfies
    \begin{equation}
        \frac{f(0)}{\hat{f}(0)} = \sup\Bigg\{\frac{f\big(\frac{x}{d}\big)}{\hat{f}(d x)} \Biggmid x \in [0,d] \Bigg\} \, .
    \end{equation}
\end{lemma}
We employ this Lemma to prove the following quantum version of the Levenshtein sphere packing bound.
\QuantumLevenshtein*

\begin{proof}
    For $d \in \R_{>0}$ the scaled Levenshtein function $f(x/d)$ and its Fourier transform $d^{2N} \hat{f}(d x)$ satisfy the conditions of Theorem~\ref{thm:quantum_cohn_elkies}. Applying the theorem, we have
    \begin{equation}
    \begin{split}  
    K &\leq \frac{1}{d^{2N}(1-\epsilon)} \sup\Bigg\{\frac{f\big(\frac{x}{d}\big)}{ \hat{f}(d x)} \Biggmid x \in [0,d] \Bigg\} \, \\
    & = \frac{1}{d^{2N}(1-\epsilon)} \frac{f(0)}{\hat{f}(0)}\, .
    \end{split}
    \end{equation}
   The equality is a consequence of Lemma~\ref{lem:levenshtein_supremum} and the result follows from the equality $f(0)/\hat{f}(0) = 1/c_{N} = j_{N}^{2N} / (N! 2^{N})$.
\end{proof}
Note that the appearance of an upper bound $d_{+}$ on the applicability of the theorem is necessary as otherwise the bound would prohibit the existence of codes of any size $K$ once the distance $d$ is taken large enough. As any pure state is a $K = 1$ code of infinite distance, however, this leads to a contradiction. The value of $d_{+}$ appearing in Theorem~\ref{thm:quantum_levenshtein} is a consequence of Lemma~\ref{lem:levenshtein_supremum} and can likely be improved substantially. 

Importantly, an improvement to $d_{+} > j_{N}/ (\sqrt{2} (N!2)^{(1/2N)})$ would imply that, for sufficiently low $\epsilon$, the distance of any $K \geq 2$ code cannot exceed $j_{N}/ (\sqrt{2} ( N!K) ^{(1/2N)})$. This distance upper bound scales as $\mathcal{O}(\sqrt{N})$ as a consequence of the asymptotic expression $j_{N} \sim N + 1.8557571\, N^{1/3} + \mathcal{O}(N^{-1/3})$ for the smallest positive Bessel zero~\cite{conway_1999_Sphere}.

\subsection{Optimality of $E_8$ and Leech GKP codes}

\begin{figure}
     \centering
     \includegraphics[width=\textwidth]{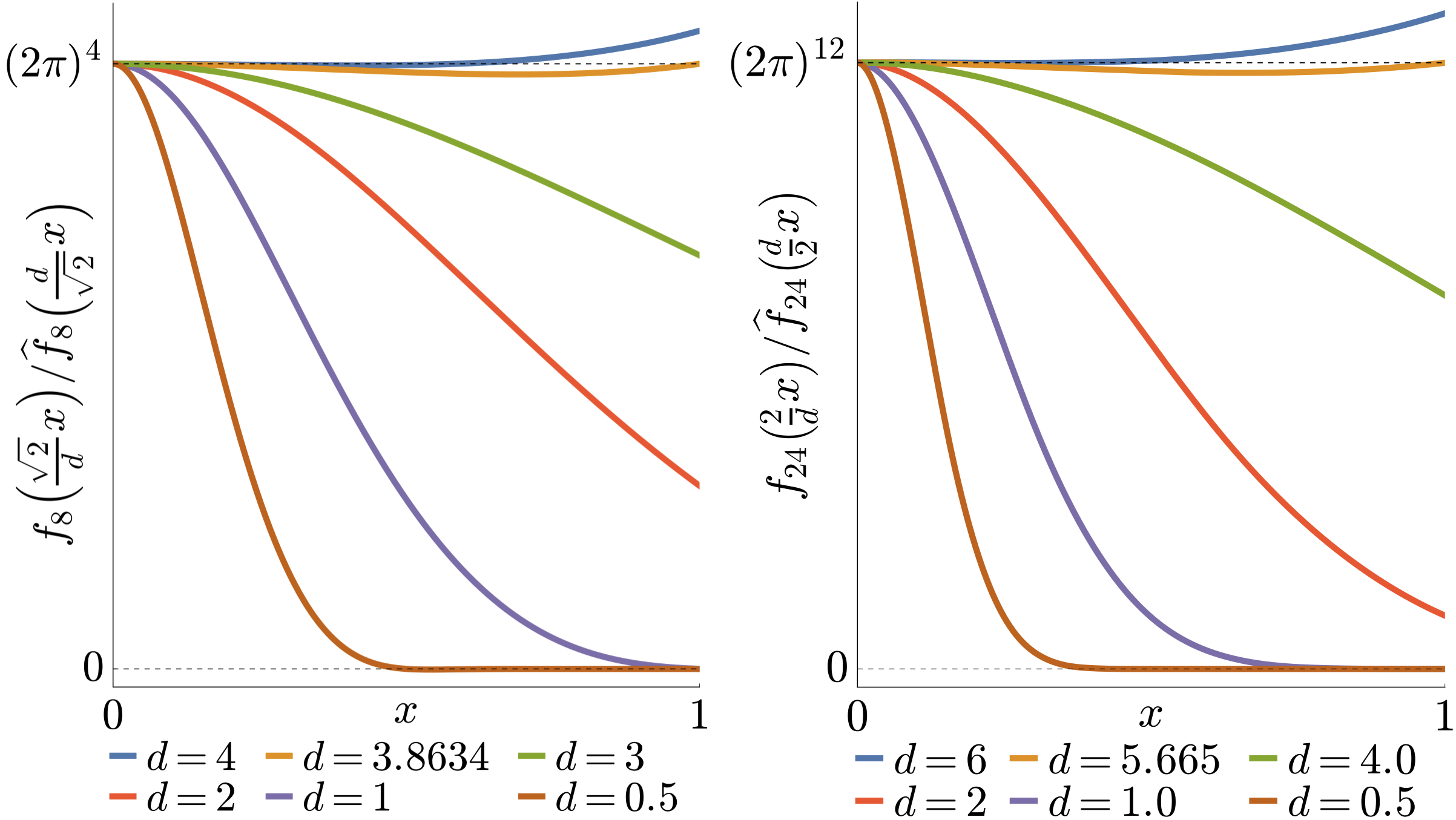}
     \caption{Left: Plot of the quotient $f_8(\sqrt{2}x ) / \hat{f}_8(d^2 x /\sqrt{2})$ over the unit interval for various values of $d$. The quotient visibly achieves a maximum value of $(2\pi)^4$ as long as $d \leq d_{8}^{\textnormal{(max)}} \approx \dEcrit$. For higher values of $d$ the quotient achieves values exceeding $(2\pi)^4$. Right: Plot of the quotient $f_{24}(2 x) / \hat{f}_{24}(d^2 x /2)$ over the unit interval for various values of $d$. The quotient visibly achieves a maximum value of $(2\pi)^{12}$ as long as $d \leq d_{24}^{\textnormal{(max)}} \approx \dLcrit$. For higher values of $d$ the quotient achieves values exceeding $(2\pi)^{12}$.}
    \label{fig:f8_quotient}
\end{figure}

Solutions to the classical sphere packing problem have, until recently, only been known in dimensions $1,2$ and $3$. In 2017 the sphere packing problem in dimension $8$ was solved by Viazovska, leading also to a subsequent solution in dimension $24$ \cite{viazovska_2017_sphere, cohn_2017_sphere}. In both dimensions the problem was solved by application of the Cohn-Elkies bound to specially constructed ‘magic’ functions $f_8$ and $f_{24}$ and establishes the $E_8$ and Leech lattices as sphere packings of maximal density. The magic functions satisfy $f_{8}(0) = f_{24}(0) = 1$, while their non-negative Fourier transforms have $\hat{f}_{8}(0) = 1/(2\pi)^4$ and $\hat{f}_{24}(0) = 1/(2\pi)^{12}$. The values $f_{8}$(x) and $f_{24}$(x) vanish exactly at arguments which are achieved as lengths of vectors within the respective lattice. That is, $f_8$ vanishes exactly whenever $x = \sqrt{2k}$ for an integer $k > 0$, the set of lengths of vectors contained in the $E_8$ lattice. At the same time $f_{24}$ vanishes for all $x = \sqrt{2k}$ with integer $k >1$,  the set of lengths of vectors contained in the Leech lattice. Both functions satisfy the properties required by Theorem~\ref{thm:quantum_cohn_elkies}. For detailed constructions of both functions see Refs.~\cite{viazovska_2017_sphere} and~\cite{cohn_2017_sphere}. Note, however, the different convention for the Fourier transformation employed in these references.

In this section, we argue that the ideal GKP codes based on the $E_8$ and Leech lattices achieve optimal distances. Our argument is based on the application of the functions $f_8$ and $f_{24}$ to Theorem~\ref{thm:quantum_cohn_elkies}. Crucially, our proof rests on the following assumptions, equivalents of Lemma~\ref{lem:levenshtein_supremum}, for both functions. 
\begin{assumption}
\label{assumption:f8_sup}
For $0 < d \leq d_{8}^{\textnormal{(max)}} \approx \dEcrit$  the $E_8$ magic function satisfies the equality 
    \begin{equation}
    \label{eq:assumption_f8}
        \textnormal{sup}\Bigg\{ \frac{f_{8}\big(\sqrt{2} x\big)}{\hat{f}_8\big( \frac{d^2}{\sqrt{2} }x\big)}\Biggmid x \in [0, 1] \Bigg\} = \frac{f_8(0)}{\hat{f}_{8}(0)} = (2 \pi)^4\, .
    \end{equation}
\end{assumption}
\begin{assumption}
\label{assumption:f24_sup}
For $0 < d \leq d_{24}^{\textnormal{(max)}} \approx \dLcrit$ the Leech magic function satisfies the equality
    \begin{equation}
    \label{eq:assumption_f24}
        \textnormal{sup}\Bigg\{ \frac{f_{24}\big(2 x\big)}{\hat{f}_{24}\big( \frac{d^2}{2}x\big)}\Biggmid x \in [0, 1] \Bigg\} = \frac{f_{24}(0)}{\hat{f}_{24}(0)} = (2 \pi)^{12}\, .
    \end{equation}
\end{assumption}
While we do not provide proofs of the equalities~(\ref{eq:assumption_f8}) and~(\ref{eq:assumption_f24}), a plot of the respective quotients, which we provide in  Fig.~\ref{fig:f8_quotient}, indicates that their suprema are achieved as assumed. 

\QuantumEEightLeechOptimality*
\begin{proof}
    By applying Theorem~\ref{thm:quantum_cohn_elkies} to the function $f(x) = f_{8}(\sqrt{2}x/d)$ with Fourier transform $\hat{f}(x) = d^8 \hat{f}_{8}(d x/\sqrt{2}) /16$ we obtain
    \begin{equation}
    \begin{split}
        K d^{8} & \leq \frac{16}{1-\epsilon} \sup\Bigg\{\frac{f_{8}\big(\frac{\sqrt{2} }{d } x \big)}{\hat{f}_{8}(\frac{d}{\sqrt{2}} x)} \Biggmid x \in [0,d] \Bigg\} \\ 
        & = \frac{16}{1-\epsilon} \sup\Bigg\{\frac{f_{8}\big(\sqrt{2}  x \big)}{\hat{f}_{8}(\frac{d^2}{\sqrt{2}} x)} \Biggmid x \in [0, 1] \Bigg\} \\
        & = (2\pi)^4\frac{16}{1-\epsilon} .
    \end{split}
    \end{equation}
    The second equality is a consequence of Assumption~\ref{assumption:f8_sup}. The second part of the theorem follows in an identical way by application of Theorem~\ref{thm:quantum_cohn_elkies} to the function $f(x) = f_{24}(2x/d)$ with Fourier transform $\hat{f}(x) = d^{24} \hat{f}_{24}(d x/2)/2^{24}$ in combination with Assumption~\ref{assumption:f24_sup}.
\end{proof}

In particular, this Theorem implies that no ideal $K \geq 2$ code on $N=4$ modes can achieve a distance exceeding $2^{7/8} \sqrt{\pi} \approx 3.251$, the distance of an ideal ($\epsilon = 0$) GKP quantum error correction code based on the $E_8$ lattice. As this is still within the range covered by the theorem, any code with larger distance must be necessarily approximate with $\epsilon > 0$. Similarly, the applicability upper bound $d_{24}^{(\textnormal{max})}$ covers the whole distance range available to ideal  error correcting codes with $N=12$, showing that no code can exceed the distance provided by an ideal Leech-lattice based GKP code. While it was previously known that ideal $E_8$ and Leech GKP codes are optimal among GKP codes~\cite{brady_2024_Advances, cohn_2009_Optimality}, our results establish that no physical code construction can achieve higher distances, even if not lattice-based. As we have seen in Section~\ref{sec:approximate}, however, the distances of approximate GKP codes are not continuous at $\epsilon = 0$. It follows that our bounds are not tight for approximate GKP codes, and it is an interesting question whether better approximate constructions can be found.

\section{Conclusion}
\label{sec:conclusion}
We have introduced weight distribution and the corresponding MacWilliams identities for operators on continuous-variable quantum systems. From these distributions we have derived a bound on general quantum error correction codes protecting against displacement noise. The bound is analogous to the classical Cohn-Elkies sphere packing bound and as such our work extends the set of previously available quantum coding bounds based on linear programming~\cite{shor_1997_Quantum, ashikhmin_1999_Upper, rains_1999_Quantum, ouyang_2022_Linear, okada_2023_Quantum}.

From the general bound have derived a quantum version of the classical Levenshtein bound on sphere packing densities, which gives a concrete upper limit on code sizes in terms of mode number as well as distance. Moreover, we have shown that the distances achieved by ideal GKP codes based on the $E_8$ and Leech lattices cannot be exceeded by any physical construction, even if not lattice based. 

It is an interesting question for future research to more fully understand the relationship between the classical sphere packing problem and continuous-variable quantum error correcting codes against the displacement channel. This could include the design of codes not based on lattices or the derivation of multi-point distance bounds in terms of semi-definite programming such as available in the classical and discrete variable quantum cases~\cite{bachoc_2010_Semidefinite, schrijver_2005_New, munne_2024_SDP, cohn_2022_Threepoint}. It is also an interesting question whether a separation in packing densities exists between classical sphere packings and their quantum analogues investigated here.

Given that discrete variable weight distributions have proven useful in studying a diverse set of topics in quantum information theory—including quantum error correction~\cite{shor_1997_Quantum,rains_1999_Quantum, ashikhmin_2000_Quantum, ashikhmin_2014_Fidelity}, magic state distillation~\cite{rall_2017_Signed, kalra_2025_Invariant}, absolutely maximally entangled states~\cite{huber_2018_Bounds, scott_2004_Multipartite}, and the robustness of entanglement to noise~\cite{miller_2023_Shor, miller_2024_Experimental}—we believe that our work provides a set of tools that will aid in investigating a similar set of problems in the continuous-variable context.

\section*{Acknowledgments} 
I am most grateful to Alexander Barg for many helpful discussions and for pointing me to many relevant points in the literature, to Jens Eisert for providing the stimulating research environment that made this project possible, and to Victor Albert for hosting me on a research stay at the University of Maryland, which contributed to the progress of this project. Moreover, I would like to thank Alexander Barg, Jens Eisert, Victor Albert, Francesco Arzani, Jonathan Conrad, Philippe Faist, Johannes Frank, Alex Townsend-Teague and Peter-Jan Derks for stimulating discussion and valuable feedback on drafts of the manuscript.
I acknowledge support from the BMBF (RealistiQ, MUNIQC-Atoms, PhoQuant, QPIC-1, and QSolid), the DFG (CRC 183, project B04, on entangled states of matter), the Munich Quantum Valley (K-8), as well as the Einstein Research Unit on quantum devices.

\printbibliography
%\bibliography{bibliography.bib}
\appendix

\section{A Bessel function identity}
\label{app:Bessel_identity}
We derive the Bessel identity~\eqref{eq:Bessel_zonal_spherical_function} by computing the Fourier transform of a certain distribution in two distinct ways. Specifically, let $f$ be the Fourier transform of a radius-$r$ spherical surface in $\R^{2N}$, that is the distribution which in radial coordinated can be expressed as $f(\boldsymbol{x}) = \delta(\norm{\boldsymbol{x}} - r)$. 
By Eq.~\eqref{eq:radial_fourier_transform_background} we can obtain the radial part of its Fourier transform explicitly as
\begin{equation}
    \hat{f}(y) = \frac{J_{N-1}(yr)r^N }{y^{N-1}}\, .
\end{equation}
Directly computing the Fourier transform in $2N$-dimensional space instead via Eq.~\eqref{eq:fourier_definition} yields
\begin{equation}
\begin{split}
    \hat{f}(\boldsymbol{y}) & = \frac{1}{(2\pi)^N} \int_{0}^{\infty} \diff x \int_{\norm{\boldsymbol{x}} = x} \diff \boldsymbol{x} \,  e^{- i \boldsymbol{y}^T \boldsymbol{x}}\delta(x-r)\\
    & = \frac{1}{(2\pi)^N} \int_{\norm{\boldsymbol{x}} = r} \diff \boldsymbol{x} \, e^{- i \boldsymbol{y}^T \Omega \boldsymbol{x}}
\end{split}
\end{equation}
where in the second line we have exploited that the symplectic form $\Omega$ is an orthogonal map. It then follows that
\begin{equation}
    J_{N-1}(\norm{\boldsymbol{y}} r) = \frac{\norm{\boldsymbol{y}}^{N-1} \hat{f}(y)}{r^N} = \frac{\norm{\boldsymbol{y}}^{N-1}}{(2\pi r)^N} \int_{\norm{\boldsymbol{x}} = r} \diff \boldsymbol{x} \, e^{- i \boldsymbol{y}^T \Omega \boldsymbol{x}}\, ,
\end{equation}
which is the desired identity.

\section{Integral over dual distribution}
\label{app:Integral_B}
Let us recall the fact that the 2-mode $\swap$ gate, defined by $\swap \ket{\psi}\otimes\ket{\varphi} = \ket{\varphi}\otimes \ket{\psi}$, satisfies the relationship $\tr \swap (A \otimes B) = \tr A B$ as well as that its decomposition into displacement operators is $\swap = \frac{1}{(2\pi)^{N}} \int \diff\bxi \, D_{\bxi} \otimes D_{\bxi}^{\dagger}$. Let now $\hat{O}$ be any trace-class operator, then
\begin{equation}
\begin{split}
    \hat{O} \otimes \id  & = \swap (\id \otimes \hat{O} ) \swap\\ & = \frac{1}{(2\pi)^{2N}} \int \diff\bxi d\boeta\,  (D_{\bxi} \otimes D_{\bxi}^{\dagger} ) (\id \otimes \hat{O})(D_{\boeta}^{\dagger} \otimes D_{\boeta})\, .
\end{split}
\end{equation}
Tracing over the first subsystem we obtain
\begin{equation}
\begin{split}
    \trb{\hat{O}} \, \id & = \frac{1}{(2\pi)^{2N}} \int \diff\bxi \diff \boeta\,  \trb{D_{\bxi} D_{\boeta}^{\dagger}} \, D_{\bxi}^{\dagger} \hat{O} D_{\bmu} \\ & = \frac{1}{(2\pi)^{N}} \int \diff\bxi\,  D_{\bxi}^{\dagger} \hat{O} D_{\bxi} \, .
\end{split}
\end{equation}
Here we have used the orthogonality of displacement operators. As a consequence we obtain the integral identity 
\begin{equation}
    \int \diff r \,  \bB(r)  = \int \diff \bxi \tr \big( D_{\bxi} \hat{O}_{1} D_{\bxi}^{\dagger} \hat{O}_{2}^{\dagger} \big)  
     = (2 \pi)^{N} \trb{\hat{O}_{1}} \trb{\hat{O}_{2}^{\dagger}}\, .
\end{equation}

\section{Finite energy GKP code parameters}
\label{app:finite_energy_GKP}
In this section we consider the code parameters of approximate, finite energy GKP states. For simplicity we will focus on single mode codes ($N=1$) encoding a single qubit ($K=2$). The projectors of ideal GKP codes can be formally expressed as
\begin{equation}
\label{eq:gkp_ideal_projector}
    \Pi = \sum_{\bxi \in \CL} e^{i \phi(\bxi)} D(\bxi) \, ,
\end{equation}
where $\CL$ denotes the symplectic lattice underlying the GKP code and the phases $\phi$ constitute a function on the lattice~\cite{conrad_2022_GottesmanKitaevPreskill}.
By displacing the code, i.e.\ $\Pi \mapsto D_{\bxi} \Pi D_{\bxi}^\dagger$, the phases $\phi(\bxi)$ can always be taken to vanish in the case of even lattices, which includes the case of a qubit encoded in a single mode~\cite{burchards_2024_Fiber}. As such a displacement does not affect the weight distributions of a code, as shown in Section~\ref{sec:AB_euclidean_invariance_properties}, or its error correcting properties, we will henceforth assume that $\phi = 0$.

To obtain a physical, finite energy version of these codes, we introduce the hermitian envelope operator $E_{\Delta} = e^{- \Delta^2 \hat{n}}$, with envelope-width parameter $\Delta > 0$, which exponentially suppresses high occupation number states. We define an approximate GKP code to be the projector onto the image of the trace class operator $E_{\Delta} \Pi E_{\Delta}$. As $\Delta \rightarrow  0$, the envelope operator approaches the identity map, and the code approaches its idealized form. 

It is known that the envelope operator has the following expansion in terms of displacement operators~\cite{royer_2020_Stabilization, matsuura_2020_Equivalence}
\begin{equation}
    E_{\Delta} = \frac{1}{2 \pi (1-e^{-\Delta^2})}\int \diff \bxi \, D_{\bxi}\, \text{exp}\Bigg\{-\frac{\norm{\bxi}^2}{4 \tanh (\Delta^2/2)}\Bigg\}  = \int \diff \bxi\, E(\bxi) D_{\bxi}\, ,
\end{equation}
where we have denoted by $E(\bxi) = \chi_{E_{\Delta}}(\bxi)/ 2\pi$ the characteristic function of the envelope operator.
Let us first see how the envelope operator transforms the characteristic function $\chi_{D_{\bsigma}^{\phantom{\dagger}}}(\bxi) \propto \delta(\bxi - \bsigma)$ of a single displacement operator.
\begin{equation}
\begin{split}
\label{eq:enveloped_displacement}
    & \chi_{E_\Delta D(\sigma)  E_\Delta^{\phantom{\dagger}}}(\bxi)  = \tr D^\dagger(\bxi) E_{\Delta} D(\bsigma) E_{\Delta} \\
    & = \int  \diff \boeta \diff \bmu \, E(\boeta) E(\bmu) \tr D(\bmu) D^\dagger(\bxi)  D(\boeta) D(\bsigma)  \\
    & = 2 \pi \int \diff \boeta \diff \bmu \, E(\boeta) E(\bmu)\, \text{exp}\bigg\{ \frac{i}{2} \Big(\bxi^T \Omega (\boeta - \bmu) + \boeta^T \Omega \bmu \Big)\bigg\} \delta(\bmu + \boeta - \bxi - \bsigma) \\
    & \propto \text{exp} \bigg\{ -\frac{1}{8} \Big( \frac{1}{T}\norm{\bxi - \bsigma}^2 - T \norm{\bxi + \bsigma}^2\Big) \bigg\} \, ,
\end{split}
\end{equation}
where the last step is simply the evaluation of a Gaussian integral and where we have introduced the shorthand $T = \tanh(\Delta^2/2)$. We can see that the envelope operator transforms the delta function $\delta(\bxi - \bsigma)$ in two distinct ways. The first term in the exponential is a transformation to a Gaussian of width of order $\sqrt{T}$ centered around $\bsigma$. As the first term is significant only when $\bxi \approx \bsigma$ the second term effects a total envelope of width of order $1/\sqrt{T}$ around the origin. In phase-space, the envelope operator for small $\Delta$ is hence approximately a combination of convolution with a narrow Gaussian and multiplication with by broad envelope. 

Let us now denote the ideal GKP codes logical displacement Pauli operators by  $\overline{Z}' = D(\bsigma_{1}^\perp)$ and $\overline{X}' = D(\bsigma_{2}^\perp)$ where we take $\bsigma_1^{\perp}, \bsigma_2^{\perp}$ to be the shortest such vectors in $\CLp$ and where w.l.o.g.\ $\bsigma_1^{\perp,T} \Omega \bsigma_2^\perp = 1$. The vectors $\bsigma_{1}^\perp$ and $\bsigma_{2}^\perp$ form a basis of $\CLp$ while the vectors $\bsigma_1 = 2\bsigma_2^\perp$ and $\bsigma_2 = 2\bsigma_1^\perp$ form a basis of $\CL$. Let us further denote by $\lambda_1^\perp$ the norm of shortest nonzero vector in $\CLp$. We can now compute the characteristic functions of the enveloped logical $\ket{\overline{0}}$ and $\ket{\overline{1}}$ states of the code. The non-enveloped states are themselves ideal GKP codes and their corresponding formal projectors  can be obtained by adjoining $\pm\overline{Z}'$ to the GKP stabilizer~\eqref{eq:gkp_ideal_projector}
\begin{equation}
\begin{split}
    \ketbra{\overline{0}}& = \sum_{a,b\in \Z} (-1)^{ab} D(a\bsigma_{1}^\perp + 2b \bsigma_{2}^\perp)\\
    \ketbra{\overline{1}} &= \sum_{a,b\in \Z} (-1)^{a(b+1)} D(a\bsigma_{1}^\perp + 2b \bsigma_{2}^\perp)\, .
\end{split}
\end{equation}
Combining with Eq.\eqref{eq:enveloped_displacement} the characteristic function of $E_{\Delta}\ket{\overline{0}}$ is then
\begin{equation}
    \chi_{E_\Delta  \ketbra{\overline{0}} E_\Delta}(\bxi) \propto \sum_{a,b \in \Z} (-1)^{ab}\text {exp} \bigg\{ -\frac{1}{8} \Big( \frac{1}{T}\norm{\bxi - \bsigma_{a,2b}^\perp}^2 - T \norm{\bxi + \bsigma_{a,2b}^\perp}^2\Big) \bigg\}\, ,
\end{equation}
where $\bsigma_{a,b}^\perp = a\sigma_{1}^\perp + b \sigma_{2}^\perp$. Consider now the region where $\norm{\bxi} \leq \lambda_1^\perp/2 - \delta$ for some $\delta > 0.$ If the envelope parameter $\Delta$ is sufficiently small, contributions to the characteristic function in this region will then be dominated by the sum term with $a = b = 0$, of order $\mathcal{O}(\text{exp}\big\{ -  (\lambda_1^\perp/2 - \delta)^2/(16 \Delta^2) \big\})$. For $(a,b)\neq (0,0)$ one has $\norm{\bxi - \bsigma_{a,b}^\perp} \geq \norm{\bsigma_{a,b}^\perp} - \norm{\bxi} \geq \lambda_{1}^\perp - (\lambda_{1}^\perp/2 - \delta) = \lambda_{1}^\perp/2 + \delta$. Hence contributions from other terms are of order $\mathcal{O}(\text{exp}\big\{ - (\lambda_1^\perp/2 + \delta)^2/(16 \Delta^2) \big\})$ and arbitrarily strongly suppressed relative to the $a = b = 0$ contribution as $\Delta \rightarrow 0$. Concretely then for $\norm{\bxi}^2 \leq \lambda_1^\perp /2 - \delta$ we have
\begin{equation}
\label{eq:enveloped_00_asymptotics}
    \chi_{E_\Delta  \ketbra{\overline{0}} E_\Delta}(\bxi) \propto \text{exp} \bigg\{ -\frac{1}{8} \Big( \frac{1}{T} - T \Big)\norm{\bxi}^2 \bigg\} + \mathcal{O}\bigg(\text{exp}\Big\{ - \frac{1}{16 \Delta^2} \Big(\frac{\lambda_1^\perp}{2} + \delta\Big)^2 \Big\}\bigg)
\end{equation}
as well as an identical result for the enveloped $\ket{\overline{1}}$ logical state
\begin{equation}
\label{eq:enveloped_11_asymptotics}
    \chi_{E_\Delta  \ketbra{\overline{1}} E_\Delta}(\bxi) \propto \text{exp} \bigg\{ -\frac{1}{8} \Big( \frac{1}{T} - T \Big)\norm{\bxi}^2 \bigg\} + \mathcal{O}\bigg(\text{exp}\Big\{ - \frac{1}{16 \Delta^2} \Big(\frac{\lambda_1^\perp}{2} + \delta\Big)^2 \Big\}\bigg)\, .
\end{equation}
Consider now the restrictions of the logical $X$ and $Y$ Pauli operators to the code space. These can be expressed as $\overline{X} \equiv \Pi \overline{X}' \Pi = \Pi \overline{X}'$ and equivalently for $\overline{Y}$. Expansions of these operators in terms of displacements are
\begin{equation}
\label{eq:X_logical_displacement_expansion}
    \overline{X}  =  \sum_{a,b \in \Z} (-1)^b D(\bsigma_{2a, 2b+1}^\perp)
\end{equation}
and
\begin{equation}
\label{eq:Y_logical_displacement_expansion}
    \overline{Y}  =  i\sum_{a,b \in \Z} (-1)^{a-b} D(\bsigma_{2a+1, 2b+1}^\perp) \, .
\end{equation}
Let us now consider the enveloped versions of the off-diagonal logical matrix elements $\ket{\overline{0}}\bra{\overline{1}} = 1/2 (\overline{X} - i \overline{Y})$ and $\ket{\overline{1}}\bra{\overline{0}} =  1/2 (\overline{X} + i \overline{Y})$. As the idealized expressions~(\ref{eq:X_logical_displacement_expansion},\ref{eq:Y_logical_displacement_expansion}) do not contain vectors shorter than $\lambda_1^{\perp}$, characteristic functions of enveloped operators are of order $\mathcal{O}(\text{exp}\big\{ - (\lambda_1^\perp/2 + \delta)^2/(16 \Delta^2) \big\})$ within the region $\norm{\bxi} \leq \lambda_1^\perp/2 - \delta$. Concretely, 
\begin{equation}
\label{eq:enveloped_10_asymptotics}
    \chi_{E_\Delta  \ket{\overline{1}}\bra{\overline{0}} E_\Delta}(\bxi) \in \mathcal{O}\bigg(\text{exp}\Big\{ - \frac{1}{16 \Delta^2} \Big(\frac{\lambda_1^\perp}{2} + \delta\Big)^2 \Big\}\bigg)\, 
\end{equation}
and 
\begin{equation}
\label{eq:enveloped_01_asymptotics}
    \chi_{E_\Delta  \ket{\overline{0}}\bra{\overline{1}} E_\Delta}(\bxi) \in \mathcal{O}\bigg(\text{exp}\Big\{ - \frac{1}{16 \Delta^2} \Big(\frac{\lambda_1^\perp}{2} + \delta\Big)^2 \Big\}\bigg)\, .
\end{equation}
As defined above, the finite-energy GKP code is $\CC = \textnormal{span}\big\{ E_{\Delta} \ket{\overline{0}}, E_{\Delta} \ket{\overline{1}}\big\}$. We can obtain an orthonormal basis $\{\ket{\overline{0}}_{\Delta}, \ket{\overline{1}}_{\Delta}\}$ for this codespace by the Gram-Schmidt procedure. First, we define the intermediate states
\begin{equation}
        \ket{\overline{0}; \Delta}  = \frac{1}{\mathcal{N}_0} E_{\Delta} \ket{\overline{0}} \quad \text{ and }\quad
        \ket{\tilde{1}; \Delta}  = \frac{1}{\mathcal{N}_1}  E_{\Delta} \ket{\overline{1}}
\end{equation}
with $\mathcal{N}_{0}$ and $\mathcal{N}_{1}$ normalization constants. These allow us to express the physical orthonormal basis for the code space as
\begin{equation}
    \ket{\overline{0}}_{\Delta}  = \ket{\overline{0}; \Delta} \quad \text{ and } \quad
    \ket{\overline{1}}_{\Delta}  = \frac{1}{\mathcal{N'}} \Big( \ket{\tilde{1};\Delta} - \langle\overline{0}; \Delta \ket{\tilde{1}; \Delta} \ket{\overline{0}; \Delta} \Big)
\end{equation}
with $\mathcal{N'}$ another normalization constant.
The GKP projector can then be expressed as
\begin{equation}
\label{eq:Pi_GKP_approx}
    \Pi_{\CC} = \ketbra{0; \Delta} + \ketbra{1;\Delta} = \frac{1}{\mathcal{N}_1} E_{\Delta} \Pi E_{\Delta} + \mathcal{D}
\end{equation}
where the correction term is
\begin{equation}
\label{eq:Pi_GKP_difference}
    \CD = c_{1} \ket{0;\Delta}\! \bra{1;\Delta} + c_{1}^{*} \ket{1;\Delta}\! \bra{0;\Delta} + d_1 \ket{0;\Delta}\! \bra{0;\Delta} + d_2 \ket{1;\Delta}\! \bra{1;\Delta}
\end{equation}
with coefficients of order
\begin{equation}
\label{eq:Pi_GKP_difference_order}
    c_{1,2}  \in \mathcal{O}\Big(\text{exp}\Big\{ - \frac{(\lambda_{1}^\perp)^2}{16\Delta^2} \Big\}\Big) \quad \text{ and }\quad d_{1,2}  \in \mathcal{O}\Big(\text{exp}\Big\{ - \frac{(\lambda_{1}^\perp)^2}{8\Delta^2} \Big\}\Big)\, .
\end{equation}
 Let us express the integrands of $\bA$ and $\bB$ directly in terms of matrix elements as follows:
\begin{align*}
\label{eq:A_integrand_enveloped_matrix_elements}
    \bA(r) =   \int_{\norm{\bxi} = r} \diff \bxi \, \bigg( & \bra{0; \Delta} D_{\bxi}^\dagger \ket{0;\Delta}\bra{0; \Delta} D_{\bxi}^\dagger \ket{0;\Delta} \\
      & + \bra{0; \Delta} D_{\bxi}^\dagger \ket{0;\Delta}\bra{1; \Delta} D_{\bxi}^\dagger \ket{1;\Delta}\\
     & + \bra{1; \Delta} D_{\bxi}^\dagger \ket{1;\Delta}\bra{0; \Delta} D_{\bxi}^\dagger \ket{0;\Delta}\\
      & + \bra{1; \Delta} D_{\bxi}^\dagger \ket{1;\Delta}\bra{1; \Delta} D_{\bxi}^\dagger \ket{1;\Delta} \\  & + \mathcal{O}\Big(\text{exp}\Big\{ - \frac{1}{8\Delta^2}\big( (\lambda_{1}^{\perp})^2 + \norm{\bxi}^2 \big)\Big\}\Big) \bigg) \\
\end{align*}
which gives, via Eqs.~\eqref{eq:enveloped_00_asymptotics} and~\eqref{eq:enveloped_11_asymptotics}, that
\begin{equation}
    \bA(r) = 8 \pi r \, \text{exp}\Big\{ - \frac{1}{4}\Big(\frac{1}{T}+T\Big)r^2 \Big\} + \mathcal{O}\Big(\text{exp}\Big\{ - \frac{1}{8\Delta^2}\big( (\lambda_{1}^{\perp})^2 + r^2 \big)\Big\}\Big)\, . \\
\end{equation}
Similarly we have from Eqs.~(\ref{eq:Pi_GKP_approx},\ref{eq:Pi_GKP_difference}, \ref{eq:Pi_GKP_difference_order}), that
\begin{equation}
\label{eq:B_integrand_enveloped_matrix_elements}
    \begin{split}
         \bB(r) =   \int_{\norm{\bxi} = r} \diff \bxi \, \bigg(  &  \bra{0; \Delta} D_{\bxi}^\dagger \ket{0;\Delta}\bra{0; \Delta} D_{\bxi}^\dagger \ket{0;\Delta} \\
        & + \bra{1; \Delta} D_{\bxi}^\dagger \ket{0;\Delta}\bra{0; \Delta} D_{\bxi}^\dagger \ket{1;\Delta}\\
        & + \bra{0; \Delta} D_{\bxi}^\dagger \ket{1;\Delta}\bra{1; \Delta} D_{\bxi}^\dagger \ket{0;\Delta}\\
        & + \bra{1; \Delta} D_{\bxi}^\dagger \ket{1;\Delta}\bra{1; \Delta} D_{\bxi}^\dagger \ket{1;\Delta} \\ & + \mathcal{O}\Big(\text{exp}\Big\{ - \frac{1}{8\Delta^2}\big( (\lambda_{1}^{\perp})^2+ \norm{\bxi}^2 - \lambda_{1}^{\perp}\norm{\bxi})\big) \Big\}\Big) \Bigg) \\
    \end{split}
\end{equation}
and hence, via Eqs.~(\ref{eq:enveloped_00_asymptotics}, \ref{eq:enveloped_11_asymptotics}, \ref{eq:enveloped_10_asymptotics}) and \eqref{eq:enveloped_01_asymptotics}, that
\begin{equation}
    \bB(r) = 4 \pi r \, \text{exp}\Big\{ - \frac{1}{4}\Big(\frac{1}{T}+T\Big) r^2 \Big\} + \mathcal{O}\Big(\text{exp}\Big\{ - \frac{1}{8\Delta^2}\big( \lambda_{1}^{\perp}- \norm{\bxi}\big)^2 \Big\}\Big)\, .
\end{equation}
We hence obtain that for $\norm{\bxi} \leq \lambda_{1}^\perp/2 - \delta$ we have the following uniform bound on the quotient
\begin{equation}
    \frac{\bA(r)}{2\bB(r)} = 1 + \mathcal{O}\Big(\text{exp}\Big\{-\frac{\lambda_{1}^\perp \delta}{8\Delta^2} \Big\}\Big)
\end{equation}
i.e.\ these approximate GKP codes are $[[N=1, K=2, \lambda_{1}^\perp/2 - \delta, \epsilon]]$-QEDCs with $\epsilon \in \mathcal{O}(\text{exp}\{ -\lambda_{1}^\perp \delta/(8 \Delta^2) \})$ for all $\delta > 0$.

\section{Proof of Lemma 2}
\label{app:levenshtein_d_proof}

We now turn to the proof of Lemma~\ref{lem:levenshtein_supremum}. Let us first consider the case $N = \frac{1}{2}$. This case is not strictly relevant to the quantum setting, where $N$ denotes the number of modes and must be an integer. However, for $N=\frac{1}{2}$ the Fourier transform of the Levenshtein function is available explicitly, which simplifies the proof of this special case. It can then serve to illustrate some of the structure of the more general proof below. In particular, for $N= \frac{1}{2}$, we have:
\begin{equation}
    f(x) = \frac{\sin^2(\pi x)}{\pi^2 x^2 (1-x^2)}
\end{equation}
and
\begin{equation}
    \hat{f}(x) = \frac{1}{4 \sqrt{2 \pi^3}} \big(4 \pi + 2 \sin{x} - 2x\big) \chi_{2\pi}(x)\, .
\end{equation}
Let us now show that for sufficiently low $d$ the quotient $f(x/d)/\hat{f}(xd)$ achieves its supremum on $[0,d]$ at $x=0$. Equivalently,
\begin{equation}
    \frac{1}{\hat{f}(0)} \geq \frac{f(x/d)}{\hat{f}(x d)}
\end{equation}
which, for $d < \sqrt{2 \pi} \approx 2.5066$, can be restated as
\begin{equation}
      1 + \frac{1}{2 \pi} \big(\sin( dx ) - d x\big) \geq \frac{\sin^2(\frac{\pi x}{d})}{ (\frac{\pi x}{d})^2 \big(1-(\frac{x}{d})^2\big)} \, .
\end{equation}
Applying the product formula:
\begin{equation}
    \frac{\sin x}{x} = \prod_{n=1}^{\infty} \Bigg( 1 - \frac{x^2}{ \pi^2 n^2} \Bigg)\, ,
\end{equation}
 the above inequality is equivalent to 
 \begin{equation}
     1 + \frac{1}{2 \pi} \big(\sin( d x) - d x\big) \geq \Big(1-\frac{x^2}{d^2}\Big) \prod_{n=2}^{\infty} \Bigg( 1 - \frac{x^2}{ d^2 n^2} \Bigg)^2 \, .
 \end{equation}
 We can further employ the inequality $\sin x \geq x - x^3/6 $, valid for positive $x$, as well as the fact that for $x \in [0, d]$ all terms in the infinite product are non-negative and bounded by $1$ to obtain the more restrictive inequality
 \begin{equation}
     1 - \frac{d^3 x^3}{12 \pi} \geq 1 - \frac{x^2}{d^2} 
 \end{equation}
 which is satisfied on $[0, d]$ if it is satisfied at $x = d$, i.e.\@ when 
 \begin{equation}
     d \leq (12 \pi)^{1/6} \approx 1.8311 \, .
 \end{equation}
 Thus we have shown that for $0 < d \leq (12 \pi)^{1/6}$ the quotient $f(x/d)/\hat{f}(xd)$ achieves its supremum on $[0,d]$ at the origin.

For general $N \in \N$ we need to establish that for sufficiently small $d$ and $x \in [0, d]$ the following inequality holds 
\begin{equation}
    \frac{{\hat{f}(x d)}}{\hat{f}(0)} \geq f\Big(\frac{x}{d}\Big) = \frac{N!^2 4^N}{(1-(x/d)^2)} \frac{J_{N}(j_N x/d)^2}{(j_{N} x/d)^{2N}}\, .
\end{equation}
Employing the following product formula for Bessel functions:
\begin{equation}
\label{eq:bessel_product_formula}
    \frac{J_N(x)}{x^{N}} = \frac{1}{2^N N!} \prod_{k=1}^{\infty} \Bigg( 1 - \frac{x^2}{j_{N,k}^2} \Bigg) \, ,
\end{equation}
where $j_{N,k}$ denotes the $k$-th positive zero of $J_N$, we find that
\begin{equation}
    \label{eq:f_levenshtein_upper_bound}
    f\Big(\frac{x}{d}\Big) = \Big(1 - \frac{x^2}{d^2} \Big)\prod_{k=2}^{\infty} \Bigg( 1 - \frac{j_N^2 x^2}{j_{N,k}^2 d^2}\Bigg) \leq \Big(1 - \frac{x^2}{d^2}\Big)\, .
\end{equation}
The last inequality holds on $[0,d]$ as all terms in the product are positive and smaller than $1$. Let us now find a lower bound on ${\hat{f}(x d)}/\hat{f}(0)$. While $\hat{f}$ is not available explicitly for general $N \in \N$, recall that it can be expressed as a convolution in $\R^{2N}$
\begin{equation}
    \hat{f} = c_{N} g * \chi_{j_{N}}\, .
\end{equation}
Because $g$ has support within the radius $j_N$ ball we have $\hat{f}(0) = c_{N} \hat{g}(0) = c_{N}$ and hence
\begin{equation}
    \frac{\hat{f}(x/d)}{\hat{f}(0)} = \Big(g * \chi_{j_{N}}\Big)\Big(\frac{x}{d}\Big)\, .
\end{equation}
We can express the convolution as 
\begin{equation}
\label{eq:g_chi_convolution_crescent}
\Big(g * \chi_{j_{N}}\Big)(y) = 1 - \int_{B_{j_N}(0) \backslash B_{j_N}(y \hat{e}_1)} d\mathbf{r} \, g(r)\, ,
\end{equation}

\begin{figure}[t]
\centering
\includegraphics[width=.75 \textwidth]{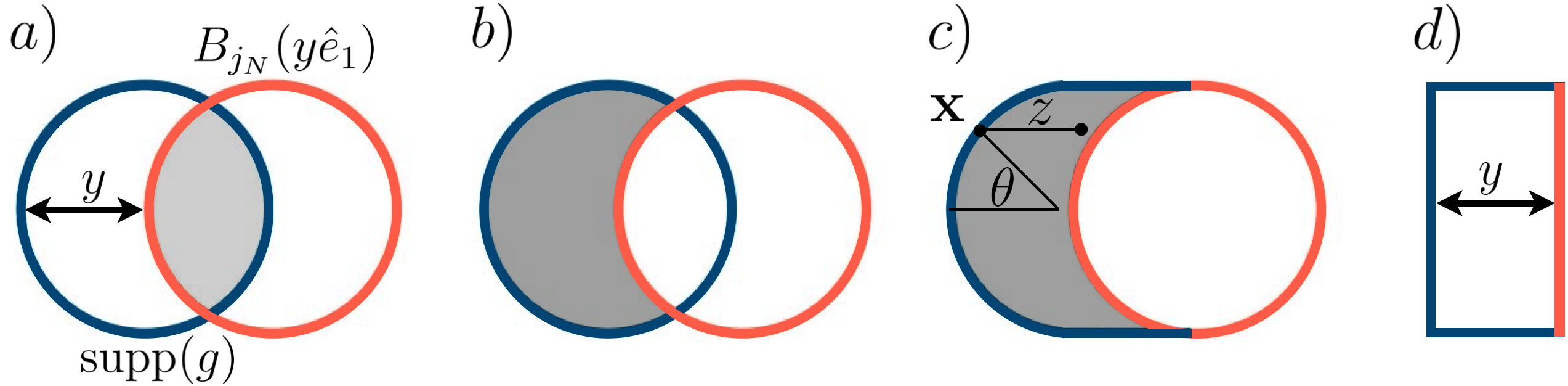}
\caption{Main steps in lower bounding the convolution product $g * \chi_{j_{N}}(y)$. a) the product is obtained given by integrating $g$ over the displaced Ball $B_{j_{N}}(y\hat{e}_1)$. As $g$ itself has support only within the radius $j_{N}$ Ball around the origin the total integration region is the intersection of two Balls (light gray). b) the integral of $g$ over $\R^{2N}$ is $1$, hence we can instead consider the integral over the complementary crescent shaped region (dark gray). c) to obtain an upper bound on the integral over the crescent shaped region we add some extra region on the top and bottom and then integrate an upper bound on $g$ over the enlarged region (dark gray). Some angles and points from the main text are also shown. d) the enlarged region can be deformed to the product of a $2N-1$-ball and the interval $[0, y]$ as it has constant width $y$ in direction $\hat{e}_{1}$.}
\label{fig:convolution_approximation}
\centering
\end{figure}
with $B_r(\mathbf{x}) \subset \R^{2N}$ the radius $r$ ball centered at $\mathbf{x}$. We will lower bound this expression by finding an upper bound to the integral. Key steps are summarized in Fig.~\ref{fig:convolution_approximation}. First, note that the radial function $g$ has the upper bound $g(r) \leq \frac{1}{2} C (r - j_N)^{2}$ for some constant $C$, to be determined later. Let now $\textbf{x} \in \partial B_{j_N}(0)$ be a point on the surface of the radius $j_{N}$ ball around the origin lying within the negative half space $x_1 < 0$. In particular let $x_1 = - j_N \cos \theta$. Then for $z \geq 0$ we have

\begin{equation}
\label{eq:g_quadratic_upper_bound}
\begin{split}
    g(\textbf{x} + z \hat{e}_1) & = g\Bigg( \sqrt{(z -j_N \cos \theta )^2 + j_N^2 \sin^2 \theta } \Bigg) \\ & = g\Bigg(\sqrt{(j_N - z \cos \theta)^2 + z^2 \sin^2 \theta}\Bigg) \\ & \leq g\big(j_N -z \cos \theta\big) \\
    & \leq \frac{1}{2} C z^2 \cos^2 \theta \, .
\end{split}
\end{equation}
where the first inequality holds since the radial component of $g$ is monotonically decaying.
Let us now denote by $\widetilde{B}(y) = \cup_{x\in [0,y]} B_{j_N}(x) \backslash B_{j_N}(y \hat{e}_1)$ the region swept out by translating the Ball at the origin by up to $y$ along the first unit vector $\hat{e}_1$ minus the Ball at $y \hat{e}_1$. Then, as $g$ is everywhere positive and $\widetilde{B}(y)$ includes the integration region in Eq.\eqref{eq:g_chi_convolution_crescent}, we have 
\begin{equation}
\Big(g * \chi_{j_{N}}\Big)(y) \geq 1 - \int_{\widetilde{B}(y)} d\mathbf{r} \, g(\mathbf{r}) \geq 1 - \frac{1}{6} C y^3 \int_{S^{(2N-1)}_{-}(j_N)} d\mathbf{r} \cos^2 \theta_{\mathbf{r}}
\end{equation}
where $S^{(2N-1)}_{-}(j_N)$ denotes the intersection of the radius $j_N$ sphere with the $x_1 < 0$ half-space and $\theta_{\mathbf{r}}$ the angle between $\mathbf{r}$ and $\hat{e}_1$. The second inequality follows from the upper bound~\eqref{eq:g_quadratic_upper_bound} and integration over the (constant) width of $\widetilde{B}(y)$. The integral can be simplified by partitioning the half-sphere into slices of lower dimensional spheres as follows
\begin{equation}
    I = \int_{S^{(2N-1)}_{-}(j_N)} d\mathbf{r} \cos^2 \theta_{\mathbf{r}} = j_N^{2N-1} \int_{0}^{\pi/2} d\theta \cos^2 \theta A^{(2N-2)}(\sin \theta)\, ,
\end{equation}
with $A^{(k-1)}(r) = \frac{2 \pi^{k/2}}{\Gamma(k/2)} r^{k-1}$ the surface area of a $k-1$-sphere of radius $r$. The remaining integral can be evaluated in terms of the Beta function $B(x,y) = \Gamma(x) \Gamma(y)/ \Gamma(x+y)$ as 
\begin{equation}
  I_k = \int_{0}^{\pi/2} d\theta\,  \cos^2 \theta \sin^{k} \theta  = \frac{1}{2}B\Big(\frac{3}{2}, \frac{k+1}{2}\Big)\, .
\end{equation}
We have thus established the following lower bound
\begin{equation}
    \hat{f}(y) \geq 1 - \frac{1}{12} C y^3 j_{N}^{2N-1} B\Big(\frac{3}{2},  \frac{2N-1}{2}\Big)
\end{equation}
and can, due to Eq.~\eqref{eq:f_levenshtein_upper_bound}, guarantee that the quotient $f(x/d)/\hat{f}(xd)$ achieves its supremum on $[0,d]$ at $x=0$ if
\begin{equation}
    1 - \frac{1}{12} C d^{3} y^3 j_{N}^{2N-1} B\Big(\frac{3}{2},  \frac{2N - 1}{2}\Big) \geq 1 - \frac{x^2}{d^2} 
\end{equation}
on $[0,d]$ which is achieved if it is achieved at $y=d$, i.e.\ when
\begin{equation}
\label{eq:d_upper_bound_C}
    d \leq \Bigg( \frac{12}{ C B\big(\frac{3}{2},  \frac{2N - 1}{2}\big) j_{N}^{2N-1} }  \Bigg)^{1/6} \, .
\end{equation}
Let us now derive a value for the constant $C$. For convenience we will introduce the following notation for the zonal spherical function $\varphi_{N}(x) = J_{N-1}(x)/x^{N-1}$. In order to establish the upper bound $g(r) \leq \frac{1}{2} C (r - j_N)^2$ note that a quadratic upper bound of the form $\varphi(x) \leq \varphi(j_{N}) + \frac{1}{2} D (x-j_{N})^2$ implies a value of  $C = c_{N} D/|\varphi(j_N)|$ when applied to Eq.~\eqref{eq:g_definition}. We can employ the identity $\partial_{x}\varphi_{N}(x) = -x \varphi_{N+1}(x)$ to obtain
\begin{equation}
\begin{split}
    \varphi_N(x) -  \varphi_N(j_N) & =  - \int_{x}^{j_N}  x' \varphi_{N+1}(x') dx' \\ &
    \leq \frac{1}{2^N N!}\int_{x}^{j_N}  x' \Big(1-\frac{x'^2}{j_N^2} \Big) dx' \\ &
    = \frac{1}{2^N N!} \Big( (x-j_N)^2 + \frac{(x-j_N)^3}{j_N} + \frac{(x-j_N)^4}{4 j_N^2}\Big) \, ,
\end{split}
\end{equation}
where the inequality follows by dropping terms from Eq.~\eqref{eq:bessel_product_formula}. This quartic upper bound can be converted into a quadratic upper bound by finding the smallest $a$ such that $a (x-j_N)^2$ upper bounds the quartic on $[0, j_N]$. We find
\begin{equation}
    \varphi_N(x) -  \varphi_N(j_N) \leq \frac{9}{ 2^{N+2} N!} (x-j_N)^2
\end{equation}
i.e. a value of $D = 9 / 2^{N+1} N!$. It follows that \begin{equation}
    C = \frac{9}{2 j_N^{N+1} |J_{N-1}(j_N)|} 
\end{equation}
and in turn, via Eq.~\eqref{eq:d_upper_bound_C}, that $f(x/d)/ \hat{f}(xd)$ achieves its maximum on $[0,d]$ at $x = 0$ for all 
\begin{equation}
    d \leq  \Bigg( \frac{16 N! |J_{N-1}(j_N)| }{3\sqrt{\pi} \Gamma\big(\frac{2N-1}{2} \big) j_{N}^{N-2}} \Bigg)^{1/6} \, .
\end{equation}

\end{document}